\documentclass[journal]{IEEEtran}

\usepackage{cite}
\usepackage{graphicx}
\usepackage{amsmath}
\usepackage{amssymb}
\usepackage{algorithm}
\usepackage{algorithmic}
\usepackage{array}
\usepackage{url}
\usepackage{booktabs}
\usepackage{multirow}
\usepackage{stfloats}
\usepackage{placeins}
\usepackage[caption=false,font=footnotesize]{subfig}
\usepackage{xcolor}
\usepackage[percent]{overpic}
\newtheorem{proposition}{\bfseries Proposition}

\UseRawInputEncoding
\begin{document}

\title{Sensing While Communicating: Active Online Spectrum Cartography via Air-Ground Cooperation}

\author{Shangjie~Zhuang, \IEEEmembership{Graduate Student Member, IEEE}, Jiahui~Liang, \IEEEmembership{Graduate Student Member, IEEE}, and~Shijian~Gao, \IEEEmembership{Member, IEEE}%
}

\markboth{Journal of \LaTeX\ Class Files,~Vol.~14, No.~8, August~2015}%
{Shell \MakeLowercase{\textit{et al.}}: Bare Demo of IEEEtran.cls for IEEE Journals}

\maketitle

\begin{abstract}
Spectrum cartography is crucial for spectrum-aware resource management in low-altitude networks, where uncrewed aerial vehicles (UAVs) collect measurements to reconstruct power spectral density (PSD) radio maps. However, limited UAV energy budgets constrain both sampling density and onboard computation, while complex urban blockages hinder reliable measurement transmission to a remote station. To address these challenges, we propose an air--ground cooperative framework for active online spectrum cartography. Through adaptive bandwidth allocation, the UAV actively collects PSD measurements using uncertainty and simultaneously transmits them to a mobile uncrewed ground vehicle (UGV), while the UGV maintains a reliable air--ground link for map updates. Specifically, we first develop an online deep-unfolded tensor decomposition method to accelerate online map updates. We then derive a bit-depth-dependent interpolation-error model to guide the selection of quantization bit depths to balance overhead and accuracy. Finally, we use uncertainty information to coordinate the UAV and UGV for active sampling and link maintenance, respectively. Extensive experiments show that the proposed reconstruction method achieves a nearly 27-fold speedup over online tensor decomposition. In urban scenes, the proposed framework outperforms the representative baselines, reducing both normalized mean squared error (NMSE) and outage ratio by more than 23\% each and increasing the packet-service completion ratio by over 4\%.
\end{abstract}


\section{Introduction}
The evolution of low-altitude wireless networks toward intelligence requires accurate radio environment awareness for coverage optimization, network planning, and spectrum management \cite{gao2026iscc}. Radio maps characterize the radio environment by associating each location with radio information \cite{zeng2024tutorial}. Received signal strength (RSS) maps were constructed to support coverage assessment and deployment planning by capturing spatial variations in aggregated received power \cite{romero2022radio}. The lack of frequency information, however, limits their use in spectrum-related applications. This limitation motivates the use of power spectral density (PSD) maps, which characterize received power jointly across space and frequency. 

Spectrum cartography refers to the reconstruction of PSD maps \cite{yang2026survey}. Early work used spline interpolation to estimate spatially varying spectral coefficients \cite{bazerque2011group}. Subsequently, block-term tensor models coupled spatial propagation fields with their corresponding spectra and characterized map recoverability \cite{zhang2020spectrum}. To accommodate off-grid sampling locations and incomplete frequency coverage in practical measurements, interpolation was integrated into tensor-based reconstruction \cite{chen2023offgrid} \cite{sun2024iibtd}. Autoencoders and conditional generative adversarial networks also exploit frequency-spatial correlations to infer maps at unmeasured frequencies \cite{zhou2023frequencyspatial}. Further studies combined neural models with structured tensor factorization \cite{shrestha2022deep}, while recent work integrates compressive sub-Nyquist sensing with neural tensor decomposition \cite{yuan2026neuralfitting}.

However, these methods primarily focus on offline reconstruction from fixed measurement sets. In low-altitude urban environments, these measurements are collected by a UAV, but its limited energy budget restricts dense sampling. Improving sampling efficiency therefore calls for selecting informative locations. This motivates active sensing, which uses uncertainty to guide subsequent trajectories~\cite{shrestha2023spectrum}. A scalable Gaussian process estimates map uncertainty to guide hierarchical informative path planning for UAV sampling in 3D environments \cite{chen2026informative}. Gaussian process ensembles and travel-cost-aware active learning select informative locations while limiting UAV travel \cite{polyzos2024bayesian}. Bayesian uncertainty and graph-based reinforcement learning have also been combined to plan informative and energy-efficient trajectories \cite{lu2025bayesian}. A building-aware 3D reconstructor and diffusion-based planner further adapt UAV trajectories to the urban environment \cite{li2026endtoend}. Although these studies improve sampling efficiency through uncertainty-guided trajectories, onboard map reconstruction is impractical for UAVs with limited energy and onboard resources. Offloading these measurements to a fixed station for reconstruction reduces the onboard computational burden, but building blockage can interrupt transmission and delay map updates. Therefore, overcoming the communication and computing bottleneck is crucial for active online spectrum cartography.

With their deployment flexibility, uncrewed ground vehicles (UGVs) have been integrated to provide communication and computing support for UAVs~\cite{11487611}. For example, in \cite{ye2025aoiairground}, UGVs dispatch UAVs from multiple stops for data collection, with multi-agent curriculum learning coordinating their actions to reduce information age. For vehicular crowdsensing, \cite{zhao2026uavcarrier} uses UGVs as mobile carriers that dispatch and recall UAVs to sense points of interest. They are coordinated by a heterogeneous learnable policy to improve data collection and geographic fairness while reducing energy consumption. When base stations are unavailable, \cite{chen2024secure} coordinates vehicle mobility and task offloading so that surveillance UAVs can offload cached tasks to UGVs with computing resources.

Moreover, measurement offloading introduces a further challenge, as sensing and communication share a fixed operating bandwidth. Online spectrum cartography needs to recover the full PSD radio map from partial frequency observations. This calls for UAV--UGV coordination to balance informative sampling and reliable communication for accurate reconstruction under the shared bandwidth, which is rarely considered in existing works.

Therefore, we propose an air--ground cooperative framework for online spectrum cartography in low-altitude urban environments: the UAV actively collects informative PSD measurements and transmits them to the UGV for online spectrum cartography. The UGV then updates an uncertainty map to select the UAV's next sampling location and guide its own repositioning to maintain reliable air--ground connectivity.

\begin{table*}
    \centering
    \caption{Feature comparison of representative radio-mapping, UAV data-collection, and UAV--UGV cooperation studies.}
    \label{tab:literature_taxonomy}
    \begingroup
    \normalsize
    \setlength{\tabcolsep}{4.0pt}
    \renewcommand{\arraystretch}{1.08}
    \begin{tabular*}{0.98\textwidth}{@{\extracolsep{\fill}}
        >{\centering\arraybackslash}m{2.0cm}
        >{\centering\arraybackslash}m{4.0cm}
        >{\centering\arraybackslash}m{1.6cm}
        >{\centering\arraybackslash}m{2.6cm}
        >{\centering\arraybackslash}m{2.3cm}
        >{\centering\arraybackslash}m{3.0cm}@{}}
        \toprule
        Refs.
         & \begin{tabular}[c]{@{}c@{}}Full-spectrum recovery\\from partial observations\end{tabular}
         & \begin{tabular}[c]{@{}c@{}}Active\\sampling\end{tabular}
         & \begin{tabular}[c]{@{}c@{}}Measurement\\quantization\end{tabular}
         & \begin{tabular}[c]{@{}c@{}}Multi-agent\\cooperation\end{tabular}
         & \begin{tabular}[c]{@{}c@{}}Communication\\resource allocation\end{tabular} \\
        \midrule
        \cite{sun2024iibtd,zhou2023frequencyspatial}
         & $\checkmark$ & $\times$     & $\times$     & $\times$     & $\times$     \\
        \cite{romero2017quantized,timilsina2024quantized}
         & $\times$     & $\times$     & $\checkmark$ & $\times$     & $\times$     \\
        \cite{polyzos2024bayesian,shrestha2023spectrum}
         & $\times$     & $\checkmark$ & $\times$     & $\times$     & $\times$     \\
        \cite{li2026endtoend}
         & $\times$     & $\checkmark$ & $\times$     & $\times$     & $\times$     \\
        \cite{li2023datacollection,chen2023shortpacket}
         & $\times$     & $\triangle$ & $\times$     & $\times$     & $\checkmark$ \\
        \cite{chen2024secure}
         & $\times$     & $\times$     & $\times$     & $\checkmark$ & $\checkmark$ \\
        This work
         & $\checkmark$ & $\checkmark$ & $\checkmark$ & $\checkmark$ & $\checkmark$ \\
        \bottomrule
    \end{tabular*}
    \vspace{0.25em}

    \begin{minipage}{0.99\textwidth}
        \footnotesize
        \textit{Legend:} $\checkmark$: explicitly considered; $\times$: outside the primary scope. Full-spectrum recovery denotes reconstruction across all target frequency bands from observations available on only a subset of bands at sampled locations. For active sampling, $\triangle$ denotes planned collection of pre-generated sensor data. Communication resource allocation includes explicit optimization of transmission time, power, bandwidth, or scheduling.
    \end{minipage}
    \endgroup
\end{table*}

Specifically, as illustrated in Fig.~\ref{fig:system_framework}, the current uncertainty guides the selection of the sensing location and frequency band. Broader sensing provides more frequency observations but increases measurement payload and leaves less bandwidth for communication. Coarser quantization reduces the payload, which can degrade reconstruction quality. To manage these trade-offs, the UAV policy jointly selects its motion, sensing and quantization to balance reconstruction accuracy and timely transmission to the UGV. However, informative sensing locations may offer poor communication conditions because of urban building blockage. Accordingly, the support planner repositions the UGV to improve air--ground connectivity and support measurement delivery. At the UGV, fully delivered measurements are used to update the PSD map and uncertainty. In turn, the uncertainty guides the UAV's subsequent sensing and mobility decisions.

Table~\ref{tab:literature_taxonomy} summarizes the key differences between the proposed framework and representative studies. The contributions of this work are summarized as follows:

\begin{itemize}
    \item We propose an air--ground cooperative framework for active online spectrum cartography in low-altitude urban environments. It integrates active UAV sensing and transmission with cooperative communication and computing support from the UGV, achieving high reconstruction accuracy and reliable transmission.
    \item A deep-unfolded online tensor decomposition algorithm is developed for sequentially delivered measurements. It replaces the computational bottleneck with a learned module and unfolds alternating updates into fixed stages, enabling timely map updates while preserving the structure for error analysis.
    \item An interpolation-error model is derived to link quantization to reconstruction. At moderate-to-high bit depths, interpolation error is approximately linear with quantization noise. This relationship guides quantization bit selection and balances the trade-off between accuracy and overhead under the UAV's fixed bandwidth.
    \item Extensive evaluations in low-altitude urban environments show that the proposed framework achieves lower estimation error, a lower outage ratio, and a higher packet-service completion ratio than representative baselines. These results reflect the effectiveness of air--ground cooperation in supporting reliable delivery and accurate reconstruction.
\end{itemize}

The remainder of this paper is organized as follows. Section~II presents the system model and problem formulation. Sections~III--V develop the online spectrum cartography method, quantization model, and air--ground cooperative framework, respectively. Section~VI presents the experimental results, and Section~VII concludes the paper.

\textit{Notation:} Plain italic letters denote scalars, while bold lowercase
and bold uppercase letters denote vectors and matrices, respectively.
Bold and nonbold calligraphic uppercase letters denote tensors and sets,
respectively.
Vector and matrix entries are denoted by $[\mathbf a]_i$ and
$[\mathbf A]_{ij}$, respectively. The quantity $\|\cdot\|_0$ counts nonzero
entries, while $|\mathcal A|$ gives the cardinality of set $\mathcal A$.
The operator $(\cdot)^\top$ denotes transpose; $\|\cdot\|_F$ and $\|\cdot\|_*$ denote
the Frobenius and nuclear norms, respectively. The symbols $\circ$ and
$\odot$ denote outer and Hadamard products, respectively.
The Gaussian distribution $\mathcal N(\mu,\sigma^2)$ has mean $\mu$ and
variance $\sigma^2$. The operator $\lfloor\cdot\rceil$ rounds to the nearest
integer. The indicator $\mathbb I\{\cdot\}$ equals 1 when the enclosed
condition holds and 0 otherwise.

\begin{figure*}[!t]
    \centering
    \includegraphics[width=\textwidth]{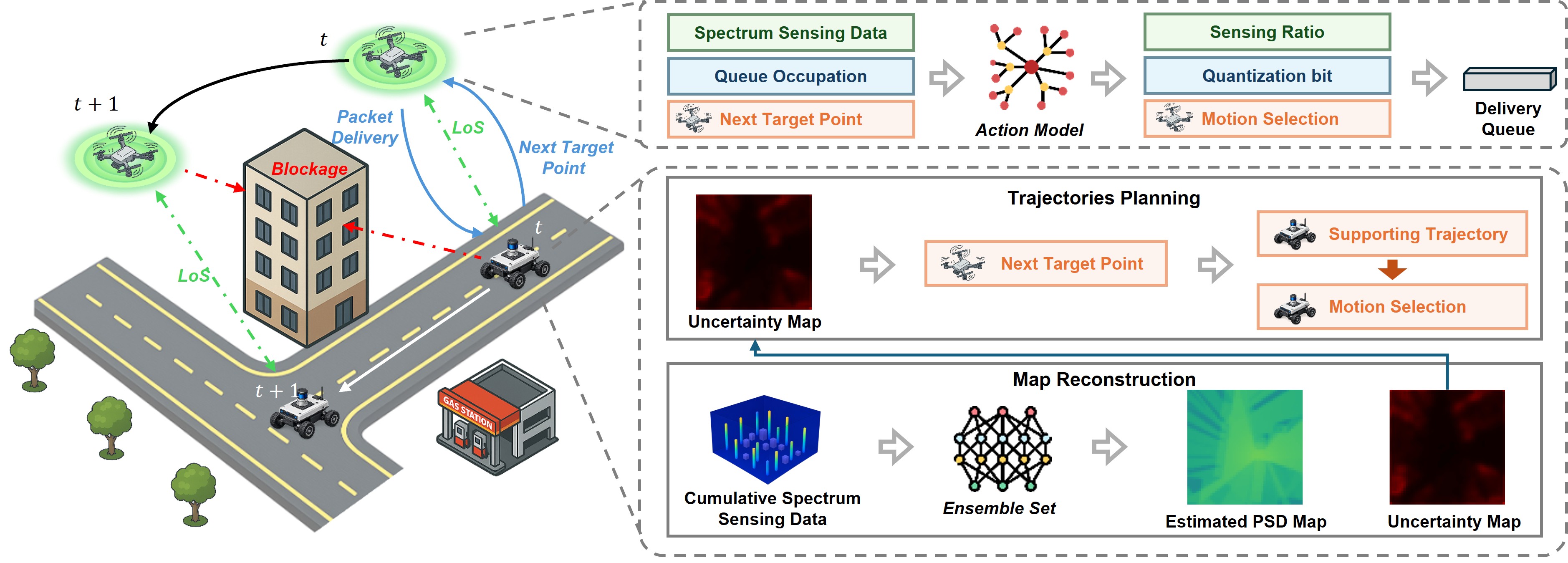}
    \caption{Closed-loop information flow of the proposed air--ground cooperative framework for active online spectrum cartography.}
    \label{fig:system_framework}
\end{figure*}

\section{System Model and Problem Formulation}

This section presents the system model for air--ground cooperative spectrum cartography. The active online spectrum cartography is formulated as an optimization problem under communication and energy constraints.

\subsection{PSD Radio Map Model}

We discretize the horizontal $x$--$y$ plane into $N_x\times N_y$ grid cells indexed by $\mathcal{G}=\{(i,j)\mid 1\leq i \leq N_x,1\leq j \leq N_y\}$. Each grid cell is represented by its center at coordinates $(i,j)$ in grid units. The monitored spectrum is divided into $K$ frequency bands indexed by $k=1,\ldots,K$. The PSD radio map is represented by a tensor $\boldsymbol{\mathcal{H}}\in\mathbb{R}_{+}^{N_x\times N_y\times K}$, where $\boldsymbol{\mathcal{H}}(i,j,k)$ denotes the received PSD at grid cell $(i,j)$ and frequency band $k$. The radio environment is generated by $R$ emitters located at grid coordinates $\boldsymbol{s}_r\in\mathcal{G}$. For source $r$, define $\mathbf{S}_r\in\mathbb{R}^{N_x\times N_y}$ as its large-scale spatial propagation field. At grid cell $(i,j)$, this field is modeled as
\begin{equation}
    [\mathbf{S}_r]_{ij}
    =
    g_r\!\left[
        d\!\left(\boldsymbol{s}_r,(i,j)\right)
    \right]
    +\zeta_r(i,j),
    \label{eq:large_scale_field}
\end{equation}
where $d(\cdot,\cdot)$ is the Euclidean distance between two grid coordinates, $g_r(\cdot)$ accounts for path loss, and $\zeta_r(\cdot)$ represents shadow fading. The power spectrum of source $r$ is denoted by $\boldsymbol{\phi}_r=[\phi_1^{r},\ldots,\phi_K^{r}]^{\top}$. Tensor decomposition (TD) has been widely used for PSD radio map modeling \cite{zhang2020spectrum,chen2023offgrid,timilsina2024quantized,sun2024iibtd}, since it exploits spatial and spectral correlations. Accordingly, the PSD radio map tensor can be expressed as
\begin{equation}
    \boldsymbol{\mathcal{H}}
    =
    \sum_{r=1}^{R} \mathbf{S}_r\circ\boldsymbol{\phi}_r.
    \label{eq:btd_map_model}
\end{equation}
Thus, the entry at grid cell $(i,j)$ and frequency band $k$ is $\boldsymbol{\mathcal{H}}(i,j,k)=\sum_{r=1}^{R}[\mathbf{S}_r]_{ij}\phi_k^{r}$.

\subsection{Partial Spectrum Observation Model}

We model the online spectrum cartography process over $T$ time slots, each lasting $\tau$ seconds. The time slots are indexed by $t \in \mathcal{T}=\{1,\ldots,T\}$. At slot $t\in\mathcal{T}$, the UAV grid coordinate is $\boldsymbol{u}_t\in\mathcal{G}$, and its altitude is $h_{\rm UAV}$.

The UAV divides its operating bandwidth $B$ into $N$ frequency units and allocates them for sensing and communication. Each sensing unit monitors one frequency band. In slot $t$, the selected sensing ratio $\rho_t\in\mathcal{R}$ assigns $n_t^s=\lceil\rho_tN\rceil$ units to sensing and the remaining units to communication. The sensing and communication bandwidths are $B_t^s=(n_t^s/N)B$ and $B_t^c=B-B_t^s$, respectively. As illustrated in Fig.~\ref{fig:bandwidth_allocation}, increasing $\rho_t$ allocates more units to sensing, providing broader spectral coverage but increasing the packet size and reducing the communication bandwidth.

\begin{figure}[t]
    \centering
    \includegraphics[width=\columnwidth]{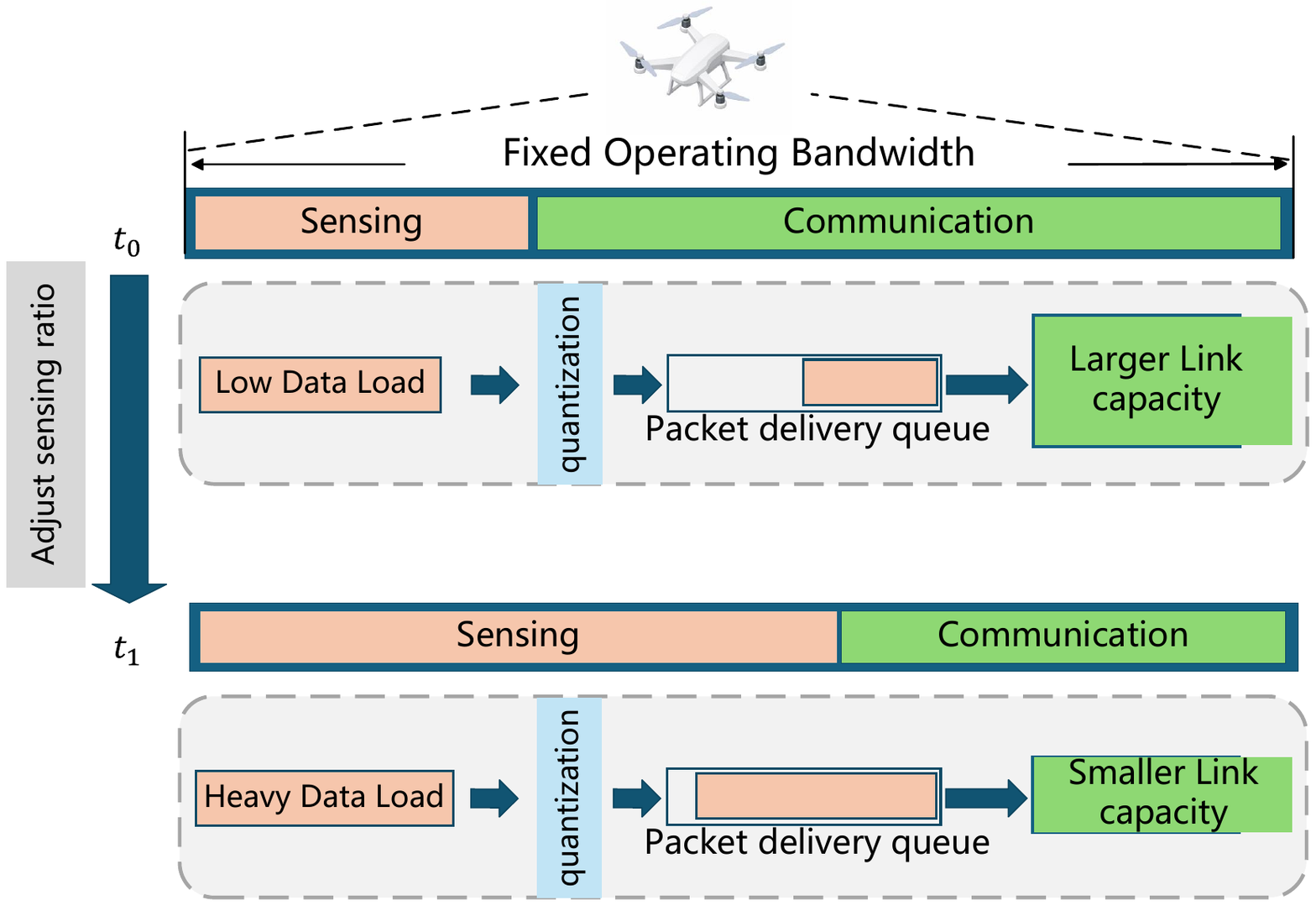}
    \caption{Slot-wise allocation of the operational bandwidth between spectrum sensing and UAV--UGV communication.}
    \label{fig:bandwidth_allocation}
\end{figure}

The UAV senses $n_t^s$ consecutive frequency bands centered on the target frequency band. Let $\boldsymbol{\omega}_t\in\{0,1\}^{K}$ denote the corresponding binary mask, where $[\boldsymbol{\omega}_t]_k=1$ if frequency band $k$ is selected and zero otherwise. Hence, $\sum_{k=1}^{K}[\boldsymbol{\omega}_t]_k=n_t^s$. During spectrum sensing, UAV measurements are affected by small-scale fading $\eta_{r,t}^{(k)}\sim\mathcal{N}(0,\sigma_\eta^2)$ and receiver noise $\epsilon_t^{(k)} \sim \mathcal{N}(0,\sigma_\epsilon^2)$.
The small-scale fading perturbations and receiver-noise samples are assumed independent across frequency bands. The aggregate observation disturbance is
$n_t^{(k)}=\sum_{r=1}^{R}\phi_k^{r}\eta_{r,t}^{(k)}+\epsilon_t^{(k)}$, and $\boldsymbol{n}_t=[n_t^{(1)},\ldots,n_t^{(K)}]^{\top}$. The sparse wideband scanning vector is
\begin{equation}
    \boldsymbol{y}_t
    =
    \boldsymbol{\omega}_t\odot
    \left[\boldsymbol{\mathcal{H}}(\boldsymbol{u}_t,:)+\boldsymbol{n}_t\right],
    \label{eq:partial_observation}
\end{equation}
where $\boldsymbol{\mathcal{H}}(\boldsymbol{u}_t,:)=[\boldsymbol{\mathcal{H}}(\boldsymbol{u}_t,1),\ldots,\boldsymbol{\mathcal{H}}(\boldsymbol{u}_t,K)]^{\top}$ is the full PSD vector at the UAV's current position.

\subsection{Communication and Queue Model}

At slot $t$, let $\boldsymbol{g}_t\in\mathcal{G}_{\rm free}$ denote the UGV grid coordinate. The set $\mathcal{G}_{\rm free}\subseteq\mathcal{G}$ contains all non-building grid cells. Building blockage determines whether the link state $\sigma_t$ is line-of-sight (LoS) or non-line-of-sight (NLoS). The three-dimensional separation is
\begin{equation}
d_t=\sqrt{
        \Delta_{\rm cell}^{2}
        \|\boldsymbol{u}_t-\boldsymbol{g}_t\|_2^{2}
        +(h_{\rm UAV}-h_{\rm UGV})^{2}
    },
\end{equation}
where $\Delta_{\rm cell}$ converts the horizontal grid distance to metres, and $h_{\rm UGV}=0$ corresponds to ground level. The service capacity in slot $t$ is
\begin{equation}
    \begin{aligned}
    C_t={}&\tau B_t^{\rm c}\log_2\!\left(1+10^{\Upsilon_t^{\rm dB}/10}\right)
    \mathbb{I}\!\left\{\Upsilon_t^{\rm dB}
    \geq\Upsilon_{\rm out}^{\rm dB}\right\},
    \end{aligned}
    \label{eq:channel_capacity}
\end{equation}
where $\Upsilon_t^{\rm dB}$ denotes the received signal-to-noise ratio (SNR), and successful data delivery requires this SNR to meet the outage threshold $\Upsilon_{\rm out}^{\rm dB}$. In slot $t$, the UAV quantizes the sensed entries of $\boldsymbol{y}_t$ at the selected bit depth $b_t\in\mathcal{B}$ and encodes the resulting indices into a packet. The packet size is
\begin{equation}
    G_t
    =
    n_t^s G_{\rm band}\frac{b_t}{b_{\rm ref}},
    \label{eq:generated_payload}
\end{equation}
where $G_{\rm band}$ denotes the payload generated by sensing one band at the reference bit depth $b_{\rm ref}$. Let $Q_t$ denote the UAV's buffer size at the start of slot $t$, and let $Q_{\max}$ denote the buffer capacity. The new packet of size $G_t$ is added before transmission. If the buffer overflows, the oldest complete packets are dropped and the total dropped size is denoted by $G_t^{\rm drop}$. Thus, the queue evolves as
\begin{equation}
    Q_{t+1}
    =\max\left(Q_t+G_t-G_t^{\rm drop}-C_t,0\right),
    \quad \forall t\in\mathcal{T}.
    \label{eq:queue_update}
\end{equation}
Packets are served in first-in, first-out order and may require multiple slots for transmission. Let $\Delta\mathcal{D}_t$ denote the data packets that were fully transmitted to the UGV in slot $t$. The cumulative UGV-side dataset is $\mathcal{D}_{t}=\mathcal{D}_{t-1}\cup\Delta\mathcal{D}_t$, $\forall t\in\mathcal{T}$, and is used for map updates.

\subsection{UAV Energy Consumption Model}

The UAV energy model accounts for mobility, spectrum sensing, and communication. The motion energy is
\begin{equation}
    E_t^{\rm mov}
    =
    E_{\rm fly}
    \left\|\boldsymbol{u}_t-\boldsymbol{u}_{t-1}\right\|_1
    +
    E_{\rm hov}
    \mathbb{I}\!\left\{\boldsymbol{u}_t=\boldsymbol{u}_{t-1}\right\},
    \label{eq:uav_motion_energy}
\end{equation}
where $E_{\rm fly}$ is the energy consumed per traversed grid cell and $E_{\rm hov}$ is the energy consumed by hovering. According to~\cite{zhang2014sensingenergy}, sensing more bands consumes more energy for a fixed slot duration. We model this bandwidth dependence and the additional multiband-processing cost as
\begin{equation}
    E_t^{\rm sen}
    =E_{\rm sen}
    \frac{n_t^s\log_2(1+n_t^s)}
    {n_{\max}^{s}\log_2(1+n_{\max}^{s})},
    \label{eq:sensing_energy}
\end{equation}
where $n_{\max}^{s}$ is the largest admissible sensing allocation and $E_{\rm sen}$ is the corresponding sensing energy. Let $P_t^{\rm tx}$ denote the UAV transmit power in watts. The communication energy is modeled as
\begin{equation}
    E_t^{\rm com}=P_t^{\rm tx}\tau.
    \label{eq:communication_energy}
\end{equation}
The total energy consumed in slot $t$ is $E_t=E_t^{\rm mov}+E_t^{\rm sen}+E_t^{\rm com}$.
Let $E_{\max}$ denote the initial UAV energy budget.

\subsection{Problem Formulation}

We formulate air--ground cooperative online spectrum cartography as a constrained optimization problem. The formulation jointly adapts UAV sensing--delivery selection to minimize the expected normalized mean squared error (NMSE) under coupled operational constraints.

At the beginning of slot $t$, the UAV policy $\pi$ selects $\boldsymbol{a}_t^{\rm u}=(m_t^{\rm u},\rho_t,b_t)$, where $m_t^{\rm u}$ denotes the UAV move. Based on the current sensing target, the UGV planner $\mu$ uses a supportive policy to select the UGV move $m_t^{\rm g}$ along a planned path. These decisions determine the UAV and UGV positions, $\boldsymbol{u}_t$ and $\boldsymbol{g}_t$. After $T$ slots, the map estimate is $\widehat{\boldsymbol{\mathcal{H}}}_T$. The corresponding NMSE is
\begin{equation}
    \mathcal{E}_T =
    \frac{\|\widehat{\boldsymbol{\mathcal{H}}}_T-\boldsymbol{\mathcal{H}}\|_F^2}
    {\|\boldsymbol{\mathcal{H}}\|_F^2},
    \label{eq:nmse_def}
\end{equation}
where $\boldsymbol{\mathcal{H}}$ is the true PSD radio map. Accordingly, we formulate problem~$\mathbf{P0}$ as
\begin{subequations}
\label{eq:online_problem}
\begin{align}
    \min_{\pi}\quad
    &\mathbb{E}_{\pi,\mu}\!\left\{\mathcal{E}_T\right\},\\
    \mathrm{s.t.}\quad
    &\mathrm{Eq.}~\eqref{eq:queue_update},
    \quad \forall t\in\mathcal{T},
    \label{eq:online_problem_queue}\\
    &Q_t+G_t-G_t^{\rm drop}\leq Q_{\max},
    \quad \forall t\in\mathcal{T},
    \label{eq:online_problem_buffer}\\
    &\rho_t\in\mathcal{R},\quad
    b_t\in\mathcal{B},
    \quad \forall t\in\mathcal{T},
    \label{eq:online_problem_actions}\\
    &\sum_{t=1}^{T}
    \left(E_t^{\rm mov}+E_t^{\rm sen}+E_t^{\rm com}\right)
    \le E_{\max}.
    \label{eq:online_problem_energy}
\end{align}
\end{subequations}

Problem~$\mathbf{P0}$ minimizes expected NMSE by coupling spectrum sensing and measurement delivery through air--ground cooperation. Constraints~\eqref{eq:online_problem_queue}--\eqref{eq:online_problem_actions} specify queue evolution, buffer capacity, and feasible sensing and quantization decisions, while Constraint~\eqref{eq:online_problem_energy} limits cumulative UAV energy. The choices $\rho_t$ and $b_t$ jointly trade reconstruction quality against delivery: increasing $\rho_t$ senses more bands but leaves less transmission bandwidth, whereas increasing $b_t$ reduces quantization distortion but enlarges the packets.

\section{Online Spectrum Cartography}
\label{sec:odu_td}

Online spectrum cartography requires efficient map updates to support uncertainty-guided sensing. This section first reviews the offline TD method and then extends it to sequentially delivered measurements for online reconstruction. A deep-unfolded version is subsequently developed for acceleration.

\subsection{Recap of the Tensor Decomposition}

Given a fixed set $\mathcal{D}$ of $M$ measurements, each measurement containing a location $\boldsymbol{z}_m$ and corresponding PSD measurements, offline TD~\cite{sun2024iibtd} fits local propagation models to jointly estimate the power spectra and spatial propagation fields in Eq.~\eqref{eq:btd_map_model}. For source $r$, offline TD approximates the propagation field around the center of grid cell $(i,j)$ using the local polynomial $f_{ij}^{r}(\boldsymbol{z}) = \boldsymbol{x}_{ij}^{\top}(\boldsymbol{z})\boldsymbol{\theta}_{ij}^{r},$ where $\boldsymbol{x}_{ij}(\boldsymbol{z})$ contains $D_{\rm p}$ polynomial features centered at $(i,j)$. The corresponding coefficient vector is $\boldsymbol{\theta}_{ij}^{r}\in\mathbb{R}^{D_{\rm p}\times 1}$. Stacking all $R$ coefficient vectors gives $\boldsymbol{\Theta}_{ij}=[(\boldsymbol{\theta}_{ij}^{1})^{\top},\ldots,(\boldsymbol{\theta}_{ij}^{R})^{\top}]^{\top}$.
The propagation estimate of source $r$ at grid cell $(i,j)$ is $f_{ij}^{r}(i,j)=\boldsymbol{e}_r^{\top}\boldsymbol{\Theta}_{ij}\approx[\mathbf{S}_r]_{ij}$. Here, $\boldsymbol{e}_r\in\mathbb{R}^{D_{\rm p}R\times 1}$ is the selection vector whose $[(r-1)D_{\rm p}+1]$-th entry is one and all other entries are zero.

The PSD measurements are recorded in $\boldsymbol{\Gamma}\in\mathbb{R}^{M\times K}$. The binary mask $\boldsymbol{\psi}\in\{0,1\}^{M\times K}$ identifies the available frequency bands, where $[\boldsymbol{\psi}]_{m,k}=1$ if frequency band $k$ is available in measurement $m$ and zero otherwise. Let $\boldsymbol{\Phi}$ collect the PSD of $R$ sources, with its $r$-th row given by $\boldsymbol{\phi}_r^{\top}\in\mathbb{R}^{1\times K}$. Define $\mathbf{X}_{ij}=[\boldsymbol{x}_{ij}(\boldsymbol{z}_1),\ldots,\boldsymbol{x}_{ij}(\boldsymbol{z}_{M})]\in\mathbb{R}^{D_{\rm p}\times M}$. The diagonal matrix $\mathbf{Q}_{ij}\in\mathbb{R}^{M\times M}$ contains the spatial kernel weights and assigns greater influence to measurements near grid cell $(i,j)$. The spatial weights and availability mask form $\mathbf{W}_{ij} = (\mathbf{I}_K\otimes\mathbf{Q}_{ij})
    \operatorname{diag}\!\left[\operatorname{vec}(\boldsymbol{\psi})\right],$
where $\mathbf{I}_K$ is the $K\times K$ identity matrix. The corresponding regression residual is $\boldsymbol{r}_{ij} (\boldsymbol{\Theta}_{ij},\boldsymbol{\Phi})  =  \operatorname{vec}(\boldsymbol{\Gamma})-(\boldsymbol{\Phi}^{\top}\otimes\mathbf{X}_{ij}^{\top})
\boldsymbol{\Theta}_{ij}.$ The offline TD reconstruction problem~$\mathbf{P1}$ is formulated as
\begin{equation}
    \begin{aligned}
        \min_{\substack{\{\boldsymbol{\Theta}_{ij}\},\boldsymbol{\Phi}, \{\mathbf{S}_r\}}}
        &\sum_{(i,j)\in\mathcal{G}}
        \left\|\mathbf{W}_{ij}
        \boldsymbol{r}_{ij}(\boldsymbol{\Theta}_{ij},\boldsymbol{\Phi})
        \right\|_2^2 \\
        &+\nu\sum_{\substack{(i,j)\in\mathcal{G}\\ r=1,\ldots,R}}
        \left(\boldsymbol{e}_r^{\top}\boldsymbol{\Theta}_{ij}
        -[\mathbf{S}_r]_{ij}\right)^2
        +\lambda\sum_{r=1}^{R}\|\mathbf{S}_r\|_* \\
        \mathrm{s.t.}\quad
        &[\boldsymbol{\Phi}]_{r,k}\geq0,\quad \forall r,k.
    \end{aligned}
    \label{eq:iibtd_objective}
\end{equation}

The first term in the objective of problem~$\mathbf{P1}$ minimizes the weighted residuals between the PSD measurements and model predictions. The second term enforces consistency between the local estimate $\boldsymbol{e}_r^{\top}\boldsymbol{\Theta}_{ij}$ and $[\mathbf{S}_r]_{ij}$. The third term applies nuclear-norm regularization to $\mathbf{S}_r$ to promote low-rank propagation fields. Offline TD solves problem~$\mathbf{P1}$ by alternating the updates of $\{\boldsymbol{\Theta}_{ij}\}$, $\boldsymbol{\Phi}$, and $\{\mathbf{S}_r\}$ until convergence, then reconstructs $\widehat{\boldsymbol{\mathcal{H}}}=\sum_{r=1}^{R}\widehat{\mathbf{S}}_r\circ\widehat{\boldsymbol{\phi}}_r$.

\subsection{Spectrum Cartography via Online TD}

At slot $t$, these definitions are applied to the cumulative delivered set $\mathcal{D}_t$. As new packets reach the UGV, reapplying offline TD would repeatedly process $\mathcal{D}_{t-1}$ even when only a few samples arrive. The online TD method therefore reuses the previous estimates while retaining the block-alternating structure of Eq.~\eqref{eq:iibtd_objective}. Fully delivered measurements are accumulated into a batch for each reconstruction update; the estimates remain unchanged between updates. At an update, the local coefficients are recomputed only for cells affected by measurements accumulated since the previous reconstruction.

\emph{Update of $\boldsymbol{\Theta}_{ij}$:}
Let $\mathbf{M}_t^{\rm aff}\in\{0,1\}^{N_x\times N_y}$ be the affected-cell mask, where $[\mathbf{M}_t^{\rm aff}]_{ij}=1$ if the spatial-kernel support of grid cell $(i,j)$ covers at least one UAV sample in the accumulated batch. For each cell with $[\mathbf{M}_t^{\rm aff}]_{ij}=1$, online TD updates the local coefficients as
\begin{equation}
    \begin{aligned}
        \boldsymbol{\Theta}_{ij,t}
        ={}&\underset{\boldsymbol{\Theta}_{ij}}{\arg\min}\quad
        \left\|\mathbf{W}_{ij,t}\boldsymbol{r}_{ij,t}
        (\boldsymbol{\Theta}_{ij},\boldsymbol{\Phi}_{t-1})\right\|_2^2\\[-1mm]
        &+\nu\sum_{r=1}^{R}
        \left(\boldsymbol{e}_r^{\top}\boldsymbol{\Theta}_{ij}
        -[\mathbf{S}_{r,t-1}]_{ij}\right)^2.
    \end{aligned}
    \label{eq:seq_theta_update}
\end{equation}
For cells with $[\mathbf{M}_t^{\rm aff}]_{ij}=0$, online TD retains $\boldsymbol{\Theta}_{ij,t}=\boldsymbol{\Theta}_{ij,t-1}$. Reusing these coefficients avoids redundant computation and improves the efficiency of online updates.

\emph{Update of $\boldsymbol{\Phi}$:}
Because $\boldsymbol{\Phi}$ is shared across grid cells, online TD initializes it with $\boldsymbol{\Phi}_{t-1}$ and solves
\begin{equation}
    \begin{aligned}
        \boldsymbol{\Phi}_{t}
        =\underset{\boldsymbol{\Phi}\geq\mathbf{0}}{\arg\min}\quad
        &\sum_{(i,j)\in\mathcal{G}}
        \left\|\mathbf{W}_{ij,t}
        \boldsymbol{r}_{ij,t}
        (\boldsymbol{\Theta}_{ij,t},\boldsymbol{\Phi})
        \right\|_2^2.
    \end{aligned}
    \label{eq:seq_phi_update}
\end{equation}

\emph{Update of $\mathbf{S}_r$:}
Online TD forms $\boldsymbol{\Psi}_{r,t}\in\mathbb{R}^{N_x\times N_y}$ with $[\boldsymbol{\Psi}_{r,t}]_{ij}=\boldsymbol{e}_r^{\top}\boldsymbol{\Theta}_{ij,t}$ and updates $\mathbf{S}_r$ by solving the low-rank subproblem
\begin{equation}
    \mathbf{S}_{r,t}
    =
    \underset{\mathbf{S}_r}{\arg\min}
    \nu\sum_{(i,j)\in\mathcal{G}}
    \left([\boldsymbol{\Psi}_{r,t}]_{ij}-[\mathbf{S}_r]_{ij}\right)^2
    +\lambda\|\mathbf{S}_r\|_*.
    \label{eq:iibtd_sr_problem}
\end{equation}
This subproblem is solved by iterative singular value thresholding (SVT) initialized at $\mathbf{S}_{r,t-1}$. The warm start reuses the previous solution and reduces the number of iterations required for convergence.

To make the computational complexity differences more explicit, we consider square grids with $N_x=N_y=N_{\rm g}$ in the following analysis. Within each alternating iteration, the local coefficients $\boldsymbol{\Theta}_{ij}$ are updated only at affected grid cells by solving Eq.~\eqref{eq:seq_theta_update}. The shared spectral matrix $\boldsymbol{\Phi}$ is updated using residual information aggregated across the grid, as specified in Eq.~\eqref{eq:seq_phi_update}. For fixed $R$, $K$, $D_{\rm p}$, and a fixed iteration budget, these updates cost $\mathcal{O}(\|\mathbf{M}_t^{\rm aff}\|_0|\mathcal{D}_t|)$ and $\mathcal{O}(N_{\rm g}^2|\mathcal{D}_t|)$, respectively. In contrast, the update of $\mathbf{S}_r$ requires repeated dense singular value decomposition (SVD) operations. With $J_{\rm svt}$ iterations, its cost is $\mathcal{O}(RJ_{\rm svt}N_{\rm g}^3)$. Overall, the updates of $\boldsymbol{\Theta}$ and $\boldsymbol{\Phi}$ are relatively inexpensive, while repeated SVD operations make the update of $\mathbf{S}_r$ the main computational bottleneck.

\subsection{Deep-Unfolded Online Spectrum Cartography}

Deep unfolding converts a number of model-based iterations into a trainable sequence of network stages while preserving the algorithmic structure~\cite{monga2021algorithm}. In radio map estimation \cite{johnson2026factor}, it unfolds a factor-decomposed convex recovery algorithm into learnable blocks, improving reconstruction performance. Online TD retains an alternating-update structure, making it naturally suited to deep unfolding. Therefore, we adopt this approach to accelerate online TD. The preceding complexity analysis identifies the update of $\mathbf{S}_r$ as the dominant computational bottleneck. Accordingly, we unfold the alternating updates into a fixed number of stages and replace iterative SVT with a learned module, yielding online deep-unfolded tensor decomposition (ODU-TD). Our selective adaptation improves computational efficiency and estimation accuracy through learned updates. Moreover, retaining model-based local interpolation and spectral estimation preserves explicit relationships among measurements, source spectra, and local estimates. These structures provide the local analytical structure for interpolation-error analysis.

\begin{figure}[!t]
    \centering
    \includegraphics[width=0.92\columnwidth]{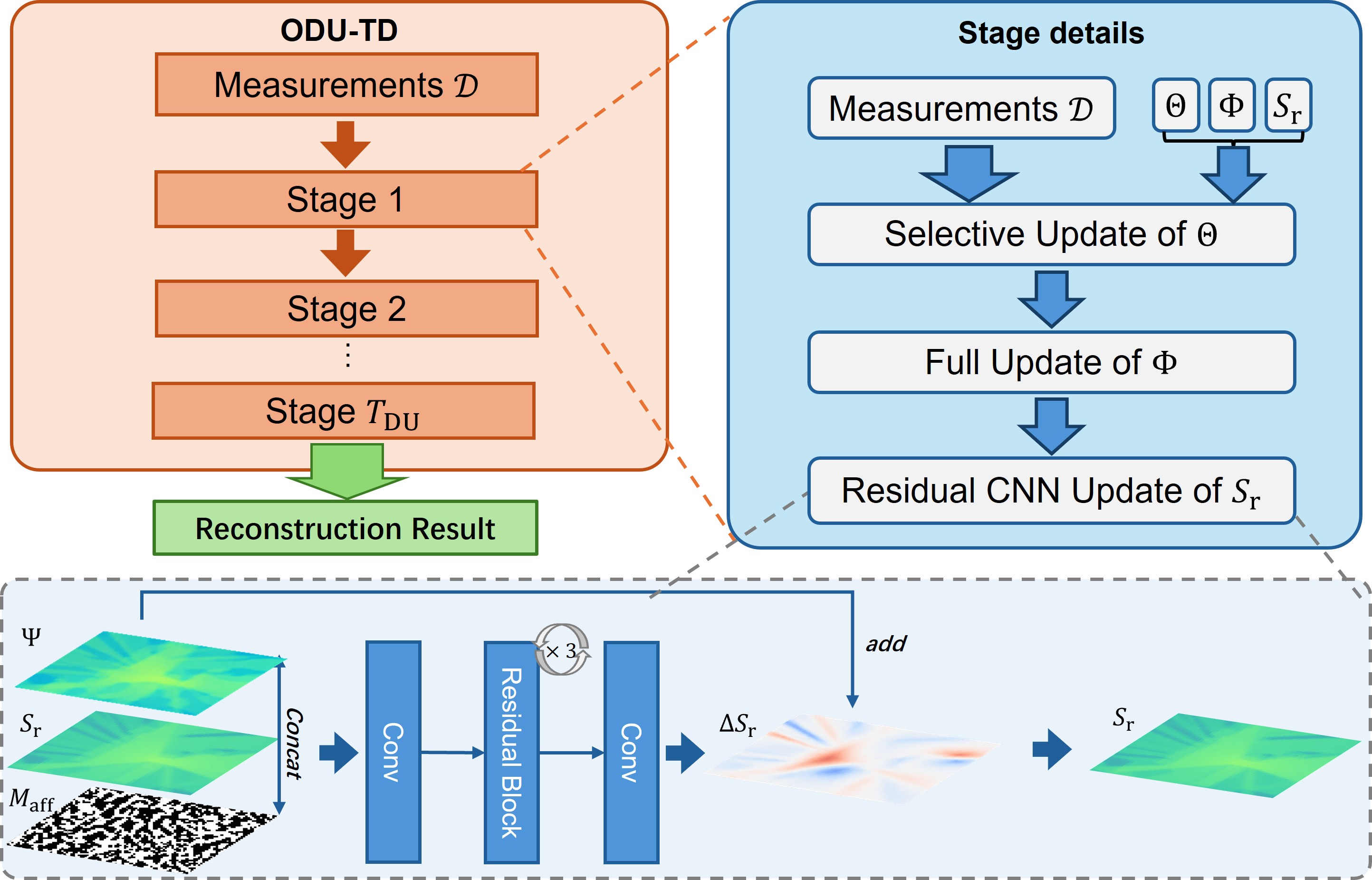}
    \caption{Architecture of the proposed ODU-TD network.}
    \label{fig:odutd_architecture}
\end{figure}

As illustrated in Fig.~\ref{fig:odutd_architecture}, ODU-TD uses a convolutional residual network to refine the spatial propagation fields. The convolutional residual network concatenates three inputs: the updated interpolation $\boldsymbol{\Psi}_r^{(\ell+1)}$, the previous-stage source map $\mathbf{S}_r^{(\ell)}$, and the affected-cell mask $\mathbf{M}_t^{\rm aff}$. This combines new information with accumulated spatial structure. Its convolutions capture neighborhood-dependent propagation patterns beyond global nuclear-norm regularization. An input convolution fuses them, residual blocks aggregate neighboring information, and an output convolution produces a residual
$\Delta\mathbf{S}_r^{(\ell+1)}=\operatorname{Prox}_{\vartheta_\ell}\{[\boldsymbol{\Psi}_r^{(\ell+1)},\mathbf{S}_r^{(\ell)},\mathbf{M}_t^{\rm aff}]\}$. The updated spatial propagation map is
\begin{equation}
    \begin{aligned}
        \bigl[\mathbf{S}_r^{(\ell+1)}\bigr]_{ij}
        ={}&\ln\!\left(
            1+e^{
            \bigl[\boldsymbol{\Psi}_r^{(\ell+1)}\bigr]_{ij}
            +\alpha_\ell\bigl[\Delta\mathbf{S}_r^{(\ell+1)}\bigr]_{ij}
            }
        \right),\\
        &\hspace{35mm}(i,j)\in\mathcal{G},
    \end{aligned}
    \label{eq:du_sr_update}
\end{equation}
where $\alpha_\ell$ scales the correction to the model-based estimate and Softplus enforces nonnegativity. 

Training combines map, source-map, and observation-level supervision through
\begin{equation}
    \mathcal{L}
    =
    \mathcal{L}_{\rm err}
    +\lambda_S\mathcal{L}_S
    +\lambda_{\rm obs}\mathcal{L}_{\rm obs},
    \label{eq:du_loss}
\end{equation}
where $\lambda_S$ and $\lambda_{\rm obs}$ are nonnegative weights. The primary loss $\mathcal{L}_{\rm err}$ is defined as the NMSE between the reconstructed and ground-truth PSD maps. The auxiliary term $\mathcal{L}_S$ supervises individual source maps, directly constraining spatial refinement beyond the combined tensor error. Finally, $\mathcal{L}_{\rm obs}$ penalizes mismatch at observed entries, discouraging refinements that conflict with the supplied measurements.

The learned module refines each source map through local convolutions and pointwise operations. For a convolutional neural network (CNN) with fixed depth, kernel sizes, and channel widths, these operations cost $\mathcal{O}(N_{\rm g}^2)$ per source map per stage. Refining $R$ source maps over $T_{\rm DU}$ stages therefore costs $\mathcal{O}(RT_{\rm DU}N_{\rm g}^2)$. In comparison, each dense SVD required by SVT costs $\mathcal{O}(N_{\rm g}^3)$. Table~\ref{tab:odu_td_complexity} compares the costs of map updates. For fixed iteration and stage counts, replacing SVT with the learned module reduces source-map refinement from cubic to quadratic complexity in $N_{\rm g}$. The fixed stage count also removes convergence-dependent repetition of the alternating updates. Algorithm~\ref{alg:odu_td_online} summarizes the complete update process.

\begin{table}[t]
    \centering
    \refstepcounter{algorithm}
    \label{alg:odu_td_online}
    \begingroup
    \normalsize
    \setlength{\tabcolsep}{0pt}
    \renewcommand{\arraystretch}{1.0}
    \begin{tabular}{@{}>{\raggedright\arraybackslash}p{\columnwidth}@{}}
        \toprule
        \textbf{Algorithm~\thealgorithm} Online Spectrum Cartography via ODU-TD \\
        \midrule
        \textbf{Input:} Cumulative delivered set $\mathcal{D}_t$, new measurements accumulated since the previous reconstruction, previous estimates $\{\boldsymbol{\Theta}_{ij,t-1}\}$, $\boldsymbol{\Phi}_{t-1}$, $\{\mathbf{S}_{r,t-1}\}$, and trained stage parameters $\{(\vartheta_\ell,\alpha_\ell)\}_{\ell=0}^{T_{\rm DU}-1}$ \\
        \textbf{Output:} Reconstructed map $\widehat{\boldsymbol{\mathcal{H}}}_t$ and updated estimates $\{\boldsymbol{\Theta}_{ij,t}\}$, $\boldsymbol{\Phi}_t$, $\{\mathbf{S}_{r,t}\}$ \\
        \begin{minipage}{\linewidth}
            \begin{algorithmic}[1]
                \STATE Form $\mathbf{M}^{\rm aff}_t$ from the sample locations in the accumulated batch
                \STATE Initialize $\boldsymbol{\Theta}_{ij}^{(0)}=\boldsymbol{\Theta}_{ij,t-1}$, $\boldsymbol{\Phi}^{(0)}=\boldsymbol{\Phi}_{t-1}$, and $\mathbf{S}_r^{(0)}=\mathbf{S}_{r,t-1}$
                \FOR{$\ell=0,\ldots,T_{\rm DU}-1$}
                \IF{$[\mathbf{M}^{\rm aff}_t]_{ij}=1$}
                \STATE Update $\boldsymbol{\Theta}_{ij}^{(\ell+1)}$ by Eq.~\eqref{eq:seq_theta_update}; 
                \ELSE
                \STATE set $\boldsymbol{\Theta}_{ij}^{(\ell+1)}=\boldsymbol{\Theta}_{ij}^{(\ell)}$
                \ENDIF
                \STATE Update $\boldsymbol{\Phi}^{(\ell+1)}$ by Eq.~\eqref{eq:seq_phi_update}
                \FOR{$r=1,\ldots,R$}
                \STATE Form $[\boldsymbol{\Psi}_r^{(\ell+1)}]_{ij}=\boldsymbol{e}_r^{\top}\boldsymbol{\Theta}_{ij}^{(\ell+1)}$
                \STATE Compute $\Delta\mathbf{S}_r^{(\ell+1)}=\operatorname{Prox}_{\vartheta_\ell}\{[\boldsymbol{\Psi}_r^{(\ell+1)},\mathbf{S}_r^{(\ell)},\mathbf{M}^{\rm aff}_t]\}$
                \STATE Update $\mathbf{S}_r^{(\ell+1)}$ by Eq.~\eqref{eq:du_sr_update}
                \ENDFOR
                \ENDFOR
                \STATE Set $\boldsymbol{\Theta}_{ij,t}=\boldsymbol{\Theta}_{ij}^{(T_{\rm DU})}$, $\boldsymbol{\Phi}_t=\boldsymbol{\Phi}^{(T_{\rm DU})}$, and $\mathbf{S}_{r,t}=\mathbf{S}_r^{(T_{\rm DU})}$
                \STATE Construct $\widehat{\boldsymbol{\mathcal{H}}}_t=\sum_{r=1}^{R}\mathbf{S}_{r,t}\circ\boldsymbol{\phi}_{r,t}$
            \end{algorithmic}
        \end{minipage} \\
        \bottomrule
    \end{tabular}
    \endgroup
\end{table}

\begin{table}[!ht]
    \centering
    \caption{Computational complexity of Online TD and ODU-TD.}
    \label{tab:odu_td_complexity}
    \begingroup
    \footnotesize
    \setlength{\tabcolsep}{3pt}
    \renewcommand{\arraystretch}{1.15}
    \begin{tabular*}{\columnwidth}{@{\extracolsep{\fill}}lcc@{}}
        \toprule
        Component & Online TD & ODU-TD \\
        \midrule
        Update passes & $J_{\rm alt}$ iterations & $T_{\rm DU}$ stages (fixed) \\
        Local fit $\boldsymbol{\Theta}$ & $\mathcal{O}(J_{\rm alt}\|\mathbf{M}_t^{\rm aff}\|_0|\mathcal{D}_t|)$ & $\mathcal{O}(T_{\rm DU}\|\mathbf{M}_t^{\rm aff}\|_0|\mathcal{D}_t|)$ \\
        Spectra $\boldsymbol{\Phi}$ & $\mathcal{O}(J_{\rm alt}N_{\rm g}^2|\mathcal{D}_t|)$ & $\mathcal{O}(T_{\rm DU}N_{\rm g}^2|\mathcal{D}_t|)$ \\
        Map refinement $\mathbf{S}_r$ & $\mathcal{O}(RJ_{\rm alt}J_{\rm svt}N_{\rm g}^3)$ & $\mathcal{O}(RT_{\rm DU}N_{\rm g}^2)$ \\
        Dense SVD count & $RJ_{\rm alt}J_{\rm svt}$ & $0$ \\
        \bottomrule
    \end{tabular*}
    \endgroup
\end{table}

\section{Measurement Quantization and Bit-Depth Selection}
\label{sec:quantization_error}
The UAV quantizes PSD measurements before sending them to the UGV to reduce the payload over the bandwidth-limited air--ground link. A lower bit depth reduces the number of transmitted bits but can increase quantization distortion and impair spectrum reconstruction. We therefore derive a bit-depth-dependent interpolation-error model to guide the trade-off between communication load and reconstruction accuracy.

\subsection{Log-Domain Quantization}

PSD measurements typically span a wide dynamic range, with most values concentrated near zero and a long tail toward higher powers~\cite{timilsina2024quantized}, as shown in Fig.~\ref{fig:psd_log_distribution}(\subref*{fig:psd_before_log}). Uniform quantization in the linear domain uses a constant absolute step across the entire range, resulting in coarse relative resolution for low-power measurements. We therefore apply a logarithmic transformation before uniform quantization to spread out these densely clustered low-power values, as illustrated in Fig.~\ref{fig:psd_log_distribution}(\subref*{fig:psd_after_log}). Equal steps in the logarithmic domain correspond to nonuniform levels in the original PSD domain, with finer absolute spacing at low powers and wider spacing at high powers~\cite{gray1998quantization}.

\begin{figure}[!ht]
    \centering
    \setlength{\abovecaptionskip}{2pt}
    \subfloat[Before log transform]{%
        \includegraphics[width=0.45\columnwidth,trim=12.78bp 14.58bp 6.84bp 4.68bp,clip]{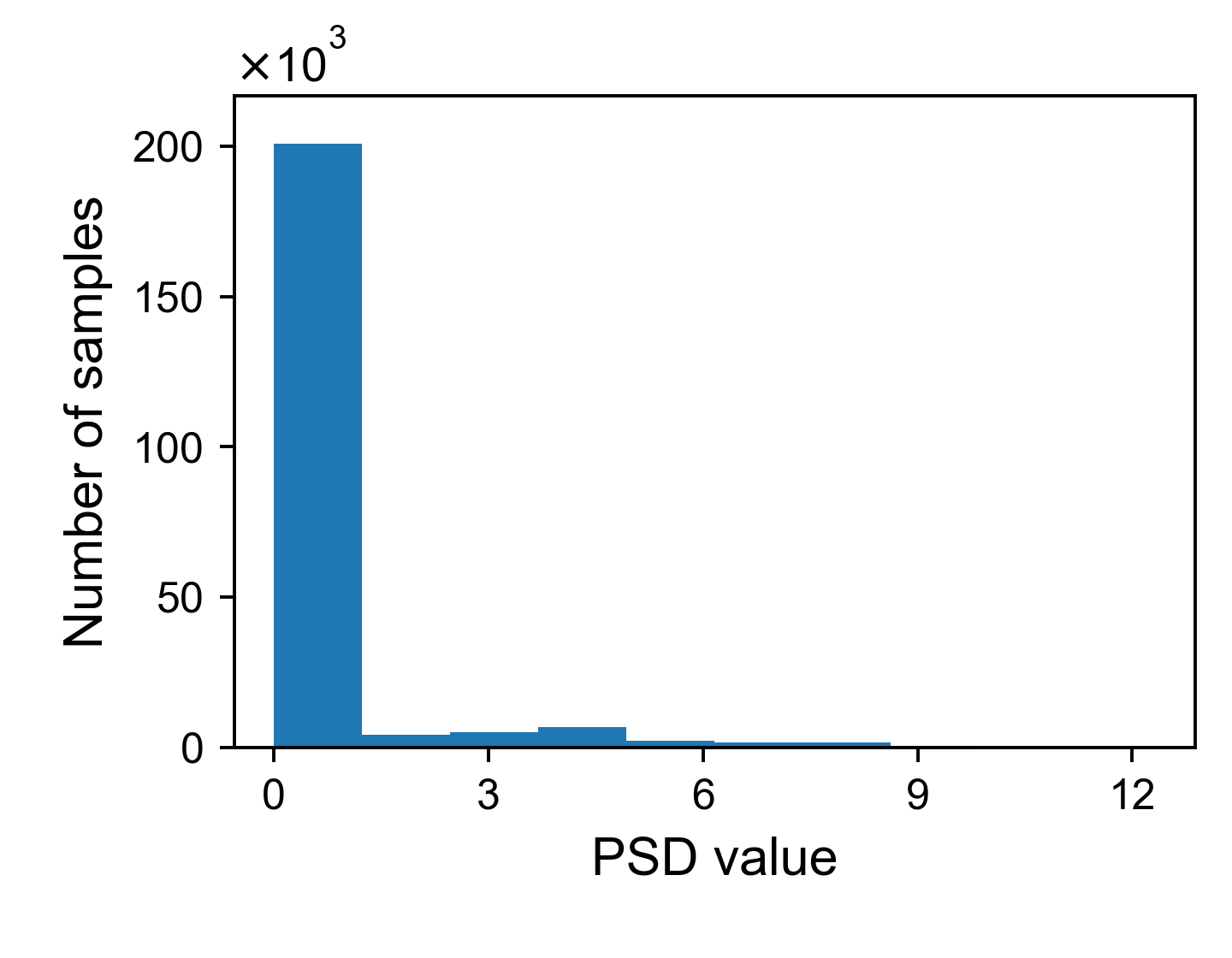}%
        \label{fig:psd_before_log}%
    }\hspace{0.012\columnwidth}%
    \subfloat[After log transform]{%
        \includegraphics[width=0.45\columnwidth,trim=12.78bp 14.58bp 6.84bp 4.68bp,clip]{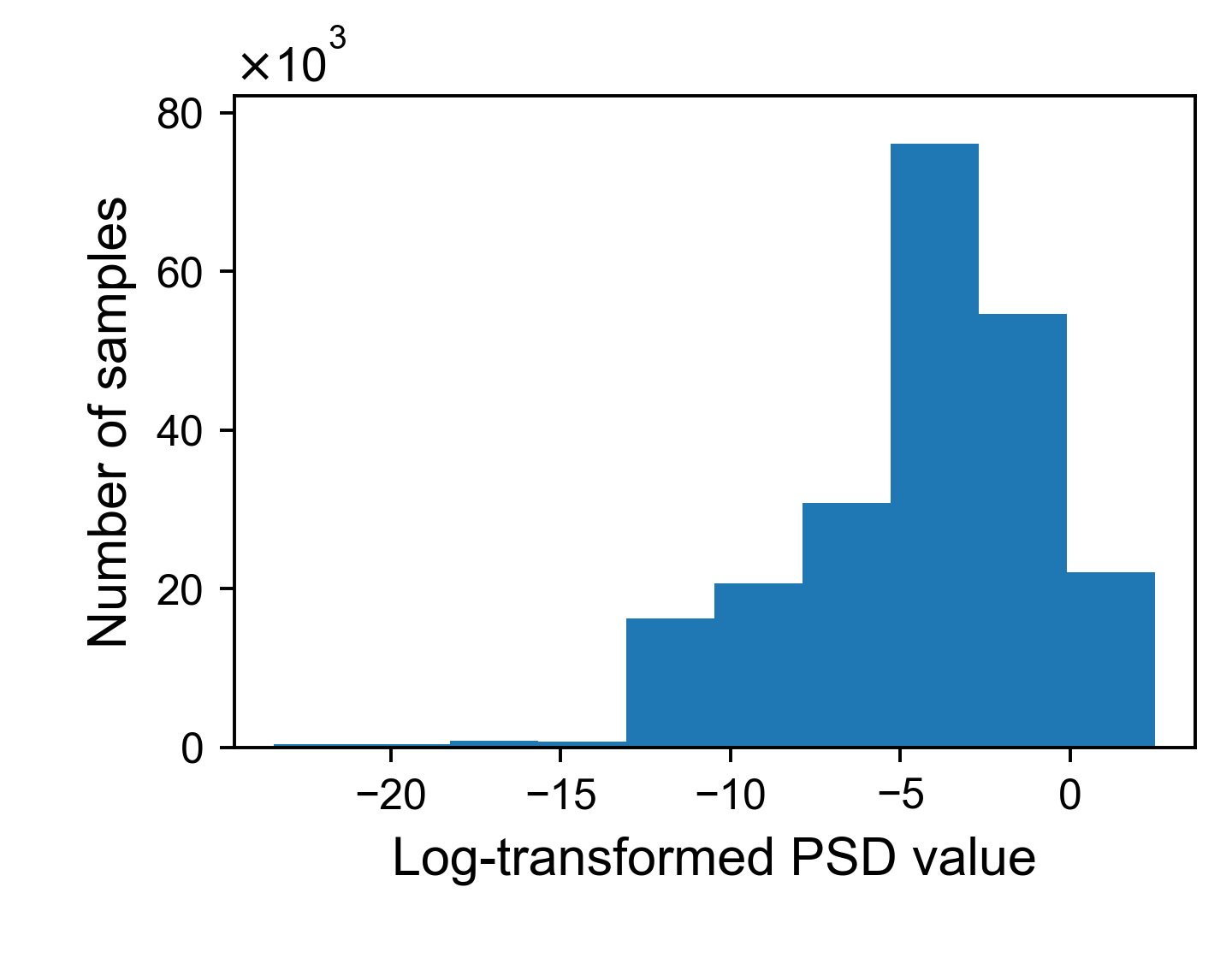}%
        \label{fig:psd_after_log}%
    }
    \caption{PSD distributions before and after the logarithmic transformation.}
    \label{fig:psd_log_distribution}
\end{figure}

To analyze how measurement quantization affects spectrum cartography performance, we apply a $b$-bit uniform quantizer after transforming the observation to the logarithmic domain. The same scalar quantization rule is applied separately to each PSD measurement, so the conditional error analysis has the same form across frequency bands. We therefore fix an arbitrary band $k$ and write $y_m\equiv y_m^{(k)}\geq0$ for the $m$-th linear-scale UAV PSD measurement, omitting $(k)$ from the associated quantization variables. Define the log-domain observation as $z_m=\ln(y_m+\epsilon_0)$, where $\epsilon_0>0$ avoids the logarithm of zero. For a given bit depth $b$, define the quantization interval over the fixed range $[z_{\min},z_{\max}]$ as $\Delta_z=(z_{\max}-z_{\min})/(2^b-1)$. After reception, the UGV obtains the dequantized log-domain observation
\begin{equation}
    \widehat z_m=z_{\min}+\Delta_z
\left\lfloor\frac{z_m-z_{\min}}{\Delta_z}\right\rceil,
\end{equation}
and the linear-domain observation is $\widehat y_m=e^{\widehat z_m}-\epsilon_0$.

We characterize the distortion using the high-resolution pseudo-quantization-noise model~\cite{widrow1996statistical,gray1998quantization}, which assumes that the log-domain quantization error is independent of the unquantized observation and uniformly distributed over one quantization interval. The corresponding linear-domain quantization error is $e_{q,m}=\widehat y_m-y_m$. The centered quantization error is defined as $\widetilde e_{q,m}(b) \triangleq e_{q,m} - \mathbb{E}\!\left[e_{q,m}\mid y_m\right].$

\begin{proposition}[\mdseries Equivalent linear-domain additive noise representation]
    \label{prop:log_quant_error}
    Under the high-resolution pseudo-quantization-noise model, the centered quantization error $\widetilde e_{q,m}(b)$ satisfies
    \begin{subequations}
        \begin{align}
            \mathbb{E}\!\left[
                \widetilde e_{q,m}(b)\mid y_m
            \right]
            &=0,
            \label{eq:equivalent_quant_noise_mean}\\
            \operatorname{Var}\!\left[
                \widetilde e_{q,m}(b)\mid y_m
            \right]
            &=
            \left(y_m+\epsilon_0\right)^2
            \left\{
                \frac{\Delta_z^2}{12}
                +O\!\left[\Delta_z^4\right]
            \right\}.
            \label{eq:prop1_log_error_variance_scaling}
        \end{align}
    \end{subequations}
\end{proposition}

\begin{IEEEproof}
See Appendix~\ref{app:proof_log_quant_error}.
\end{IEEEproof}

\textbf{Remark:} Proposition~\ref{prop:log_quant_error} provides the equivalent noise model needed to relate quantization bit depth to local interpolation-error variance. Since quantization is performed in the logarithmic domain whereas the TD estimator operates on linear-scale measurements, the uniform log-domain quantization error cannot be directly used in the subsequent interpolation-error analysis. After the inverse logarithmic transformation, it becomes a linear-domain distortion consisting of a conditional bias $\mathbb{E}\!\left[e_{q,m}\mid y_m\right]$ and a centered quantization error $\widetilde e_{q,m}(b)$. The centered error can be directly incorporated into the interpolation-error analysis, while the bias does not affect the variance calculation.

\subsection{Reconstruction-Guided Bit-Depth Selection}

To identify suitable online bit depths, we examine how the quantization noise characterized in Proposition~\ref{prop:log_quant_error} affects the local TD estimate. Recall that $\boldsymbol{e}_r^{\top}\boldsymbol{\Theta}_{ij}$ estimates the propagation-field value of source $r$ at grid cell $(i,j)$. For measurements quantized at bit depth $b$, let $\widehat{\boldsymbol{\Theta}}_{ij}(b)$ be the resulting coefficient estimate and define the interpolation error as $e_{ij}^{r}(b)
    =
    \boldsymbol{e}_r^{\top}\widehat{\boldsymbol{\Theta}}_{ij}(b)
    -[\mathbf{S}_r]_{ij},$ 
which measures its deviation from the true propagation-field value. To isolate the effect of quantization bit depth, we analyze a local TD update with the sampling and estimation settings held fixed.

\begin{proposition}[\mdseries Bit-depth scaling of TD interpolation-error variance]
    \label{prop:quant_iibtd_error}
    Under the local TD estimator and the quantization noise described above, the interpolation-error variance has the following approximate dependence on quantization bit depth in the moderate-to-high-bit regime:
    \begin{equation}
        \operatorname{Var}\!\left[e_{ij}^{r}(b)\right]
        \approx
        a(2^b-1)^{-2}+c,
        \label{eq:high_bit_td_error_scaling}
    \end{equation}
    where $a\geq 0$ and $c$ are independent of $b$.
\end{proposition}

\begin{IEEEproof}
See Appendix~\ref{app:proof_quant_iibtd_error}.
\end{IEEEproof}

The resulting interpolation error model gives an affine relationship between local interpolation-error variance and $x=(2^b-1)^{-2}$ for moderate to high bit depths. We then test offline whether NMSE follows the same dependence on $x$. For each observation ratio, we quantize randomly sampled measurements at different bit depths, reconstruct the PSD maps from these measurements using the fixed model, and fit ${\rm NMSE}=a_{\rho}x+c_{\rho}$ in the moderate-to-high-bit regime. Fig.~\ref{fig:quant_noise_fit} shows the fitting errors across the tested quantization bit depths. Quantization at 2 and 4 bits exhibits large deviations from the fitted model, whereas the other tested bit depths yield relatively small fitting errors. To balance quantization accuracy and transmission payload, we select $\mathcal{B}=\{6,8,10\}$ as the candidate bit depths for the UAV policy.

\begin{figure}[!t]
    \centering
    \includegraphics[width=\linewidth]{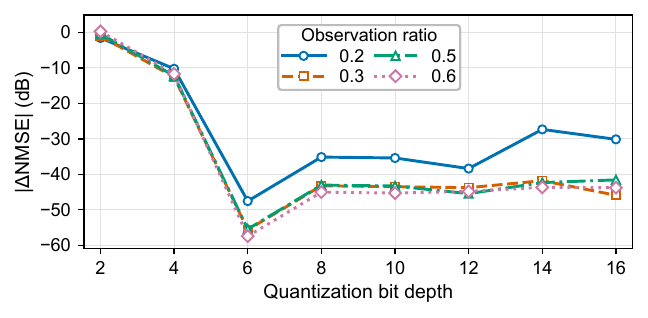}
    \caption{Deviation from the fitted NMSE model at different observation ratios, where $\Delta{\rm NMSE}={\rm NMSE}_{\rm quant}-(a_{\rho}x+c_{\rho})$.}
    \label{fig:quant_noise_fit}
\end{figure}

\section{Air--Ground Vehicle Cooperation}
\label{sec:air_ground_framework}

Urban blockage can interrupt measurement delivery, motivating the use of a mobile UGV to support communication and online reconstruction. This section integrates online reconstruction and quantization into an air--ground cooperation framework, coupling uncertainty-guided UAV sensing and delivery with UGV repositioning and support.

\subsection{Uncertainty-Guided UAV Sensing}

Ensemble methods combine predictions from multiple estimators to produce uncertainty using their disagreement. The UGV then evaluates $N_{\rm ens}$ members on the same delivered observations and each estimate is $\widehat{\boldsymbol{\mathcal{H}}}_t^{(e)}$. With $\sum_{e=1}^{N_{\rm ens}}w_{e,t}=1$, their weighted average is $\bar{\boldsymbol{\mathcal{H}}}_t=\sum_{e=1}^{N_{\rm ens}}w_{e,t}\widehat{\boldsymbol{\mathcal{H}}}_t^{(e)}$, and the resulting uncertainty is
\begin{equation}
    \boldsymbol{\mathcal{U}}_t(i,j,k)
    =
    \sum_{e=1}^{N_{\rm ens}}w_{e,t}
    \big[\widehat{\boldsymbol{\mathcal{H}}}_t^{(e)}(i,j,k)
    -\bar{\boldsymbol{\mathcal{H}}}_t(i,j,k)\big]^2.
    \label{eq:ensemble_uncertainty_map}
\end{equation}

Using the uncertainty map, the planner first averages over frequency to obtain the location-level uncertainty $[\mathbf{U}_t^{\rm sp}]_{ij}=K^{-1}\sum_{k=1}^{K}\boldsymbol{\mathcal{U}}_t(i,j,k)$. For identifying the target frequency band, the planner defines the frequency-aware score
\begin{equation}
    \boldsymbol{\mathcal{S}}_t^{\rm f}(i,j,k)
    =
    \boldsymbol{\mathcal{U}}_t(i,j,k)-\beta_f\boldsymbol{\mathcal{V}}_t(i,j,k),
    \label{eq:frequency_aware_score}
\end{equation}
where $\boldsymbol{\mathcal{V}}_t$ contains normalized visit counts and $\beta_f$ penalizes repeated frequency selections. The highest-scoring unobserved band at each candidate cell enters the joint target selection.

Let $\mathcal{C}_t^{\rm u}$ contain sampling-valid grid cells within the current search range that have at least one unobserved frequency band. For each candidate grid cell $(i,j)$, $\widehat{k}_t(i,j)$ denotes its unobserved band with the largest frequency-aware score in Eq.~\eqref{eq:frequency_aware_score}. The target location is selected as
\begin{equation}
    \boldsymbol{v}^\star
    =
    \arg\max_{(i,j)\in\mathcal{C}_t^{\rm u}}
    \left\{
    \lambda_u[\mathbf{U}_t^{\rm sp}]_{ij}
    +\boldsymbol{\mathcal{S}}_t^{\rm f}[i,j,\widehat{k}_t(i,j)]
    \right\}.
    \label{eq:uncertainty_target_selection}
\end{equation}
Here, $\lambda_u$ balances spatial and frequency-specific informativeness. Thus, $\boldsymbol{v}^\star$ and $\widehat{k}_t(\boldsymbol{v}^\star)$ specify the selected location and frequency band, respectively. To limit search complexity, $\mathcal{C}_t^{\rm u}$ is normally restricted to a configured Manhattan neighborhood.  If the relative improvement in mean local uncertainty remains below $0.08$ for two consecutive map updates, the planner temporarily searches globally and resumes local search once two nearby candidates are available.

\subsection{Predictive UGV Support Planning}
\label{sec:qpsp}

The sensing target is selected to guide the UAV, without considering whether the UGV can reach it for communication support. Simply following this target may fail to anticipate blockage along the UAV's route, delaying the UGV's response to link degradation. The UAV is responsible for sensing and measurement delivery, using a learned policy to jointly control bandwidth allocation, quantization bit depth, and movement. In contrast, the UGV has a simpler role: moving to suitable locations to provide communication and computing support. Therefore, the queue-aware predictive support planning (QPSP) planner is more suitable for UGV support than a learned policy.

The QPSP planner has two modes: predictive support and recovery, illustrated in Fig.~\ref{fig:qpsp_modes}(a) and (b), respectively. During predictive support, it predicts $x$-first and $y$-first Manhattan routes from $\boldsymbol{u}_t$ toward $\boldsymbol{v}^\star$, each limited to $H_{\rm p}$ UAV positions. Mapping these positions to the nearest traversable cells and removing duplicates yields the candidate support locations. For each candidate, the planner uses the large-scale channel model and current communication bandwidth to compute each route's link-availability ratio and mean capacity, assigning zero capacity below the outage threshold. Candidates within a specific Manhattan distance from the UGV receive priority. Within each priority group, candidates are ranked by 1) the lower link-availability ratio of the two routes, 2) the lower route-level mean capacity, and 3) the mean capacity across both routes, in that order. The highest-ranked location becomes the supporting target $\boldsymbol{q}_t^\star$.

The planner switches to recovery mode when queued data experience an outage, or a high backlog coincides with persistently insufficient service. During an outage or high backlog, it identifies the reachable ground cell nearest the current UAV position. A bounded A* search generates one ground path from the UGV toward this cell. The planner searches this path and its neighboring traversable cells for a communication support location to restore service. Once the link remains healthy and the backlog is sufficiently low, the planner resumes predictive support. In either mode, the UGV moves toward the current support goal $\boldsymbol{q}_t^\star$ along a four-neighbor A* path.

\begin{figure}[!h]
    \centering
    \setlength{\abovecaptionskip}{2pt}
    \subfloat[Predictive support]{%
        \includegraphics[width=0.485\columnwidth]{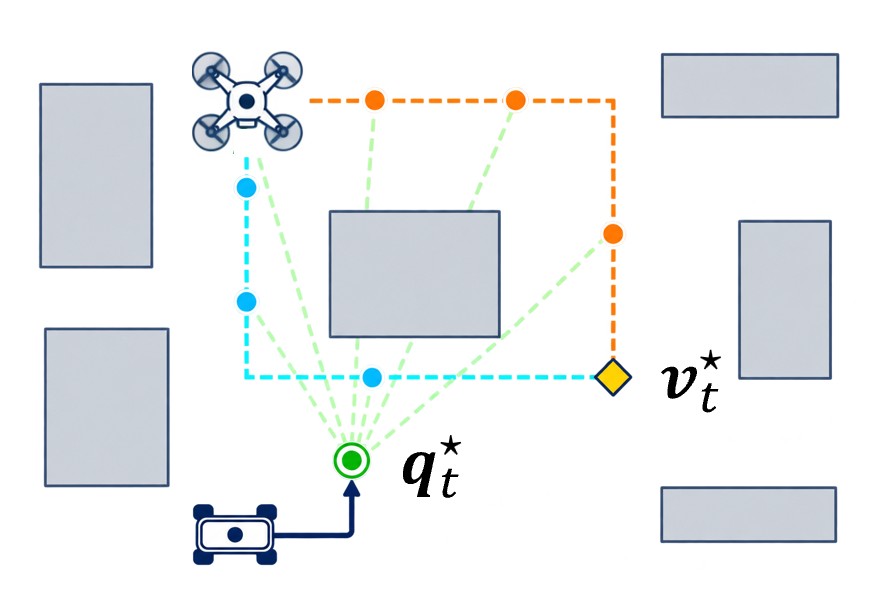}%
        \label{fig:qpsp_predictive}%
    }\hfill
    \subfloat[Recovery support]{%
        \includegraphics[width=0.485\columnwidth]{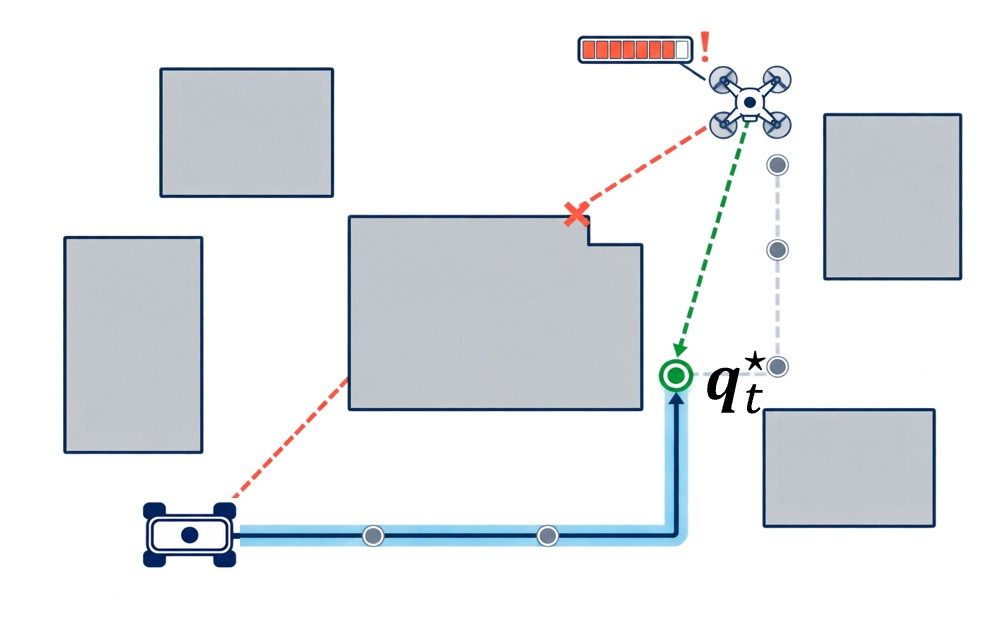}%
        \label{fig:qpsp_recovery}%
    }
    \caption{Predictive and recovery support planning for the UGV.}
    \label{fig:qpsp_modes}
\end{figure}

\subsection{UAV Policy Design and Training}
\label{sec:uav_policy_formulation}

With UGV support determined by $\mu$, the UAV policy learns to coordinate sensing and delivery from reconstruction and uncertainty information.

\textbf{1) Observation and state spaces:}
The actor observation $\boldsymbol{o}_t^{\rm u}$ combines reconstruction uncertainty and the sensing target with queue and channel states to coordinate informative sensing with measurement delivery. UAV position, remaining energy, current bandwidth allocation, and relative UAV--UGV geometry further inform movement and resource allocation. During training, the critic state $\boldsymbol{s}_t^{\rm c}$ supplements these observations with reconstruction NMSE, map-update status, and planner features to estimate the expected return.

\textbf{2) Action space:}
At slot $t$, the UAV action is defined as $\boldsymbol{a}_t^{\rm u}=(m_t^{\rm u},\rho_t,b_t)$. The movement component $m_t^{\rm u}\in\{\mathrm{stay},\mathrm{E},\mathrm{N},\mathrm{W},\mathrm{S}\}$ specifies whether the UAV remains stationary or moves east, north, west, or south. The other two components, $\rho_t\in\mathcal R$ and $b_t\in\mathcal B$, specify the sensing bandwidth ratio and quantization bit depth, respectively.

\textbf{3) Reward design:}
The reconstruction reward based on NMSE is sparse, as it is available only when the spectrum map is updated. To provide intermediate feedback and guide sensing and transmission decisions, we introduce auxiliary rewards and define the shaped reward as
\begin{equation}
    r_t
    =\lambda_{\rm N}r_t^{\rm N}
    +\lambda_{\rm Q}r_t^{\rm Q}
    +\lambda_{\rm P}r_t^{\rm P}
    +\lambda_{\rm R}r_t^{\rm R},
    \label{eq:ppo_total_reward}
\end{equation}
where $\lambda_{\rm N}$, $\lambda_{\rm Q}$, $\lambda_{\rm P}$, and $\lambda_{\rm R}$ balance the four reward components. The reconstruction reward $r_t^{\rm N}$ encourages NMSE reduction as delivered measurements update the map. The queue penalty $r_t^{\rm Q}$ discourages backlog and packet drops, promoting coordination between sensing workload and transmission service. To guide motion between map updates, $r_t^{\rm P}$ rewards progress toward the sensing target and penalizes moving away. Finally, $r_t^{\rm R}$ penalizes sensing that adds no new frequency observations at the sampled location, discouraging redundant use of sensing and transmission resources.

\textbf{4) Policy training:}
Compared with vanilla policy gradients, proximal policy optimization (PPO)~\cite{schulman2017ppo} discourages excessive policy updates through objective clipping, making it suitable for discrete UAV actions and asymmetric actor--critic training. For each update, generalized advantage estimation (GAE) uses the collected rewards and critic estimates to compute advantages and return targets. The training loss is
\begin{equation}
    \mathcal{L}(\vartheta_{\rm a},\vartheta_{\rm c})
    =
    \mathbb{E}_t\!\left[
        \mathcal{L}_t^{\rm A}(\vartheta_{\rm a})
        +c_{\rm v}\mathcal{L}_t^{\rm V}(\vartheta_{\rm c})
        +c_{\rm e}\mathcal{L}_t^{\rm E}(\vartheta_{\rm a})
    \right].
    \label{eq:ppo_training_loss}
\end{equation}
Here, $\mathcal L_t^{\rm A}$ is the negative clipped surrogate objective limiting large policy changes, and $\mathcal L_t^{\rm V}$ is the clipped value loss fitting return targets. The entropy loss $\mathcal L_t^{\rm E}$ is the negative policy entropy, whose minimization encourages exploration. The coefficients $c_{\rm v}$ and $c_{\rm e}$ weight the latter two terms. The actor and critic are optimized over minibatches of the collected interactions using this loss.
\begin{table}[!t]
    \centering
    \refstepcounter{algorithm}
    \label{alg:air_ground_training}
    \begingroup
    \normalsize
    \setlength{\tabcolsep}{0pt}
    \renewcommand{\arraystretch}{1.0}
    \begin{tabular}{@{}p{\columnwidth}@{}}
        \toprule
        \textbf{Algorithm~\thealgorithm} UAV Policy Training with UGV Support \\
        \midrule
        \textbf{Input:} Environment, pretrained ensemble, target planner, $\mu$, reconstruction batch size, $L$, $N_{\rm upd}$, $\mathcal{D}_0$ \\
        \textbf{Output:} Trained parameters $\vartheta_{\rm a}$, $\vartheta_{\rm c}$ \\
        \begin{minipage}{\linewidth}
            \begin{algorithmic}[1]
                \STATE Initialize $\vartheta_{\rm a}$, $\vartheta_{\rm c}$
                \FOR{$N_{\rm upd}$ policy updates}
                \STATE Clear the rollout buffer
                \FOR{$L$ interaction steps}
                \STATE Obtain $\boldsymbol{o}_t^{\rm u}$, $\boldsymbol{s}_t^{\rm c}$
                \STATE Sample feasible $\boldsymbol{a}_t^{\rm u}\sim\pi_{\vartheta_{\rm a}}(\cdot|\boldsymbol{o}_t^{\rm u})$
                \STATE Select $m_t^{\rm g}$ using $\mu$
                \STATE Execute UAV and UGV actions
                \STATE Update $Q_{t+1}$ by Eq.~\eqref{eq:queue_update}
                \STATE $\mathcal{D}_t\leftarrow\mathcal{D}_{t-1}\cup\Delta\mathcal{D}_t$
                \STATE Append $\Delta\mathcal{D}_t$ to the pending batch
                \IF{the reconstruction batch is ready}
                \STATE Update estimated maps by Algorithm~\ref{alg:odu_td_online}
               \STATE Compute $\boldsymbol{\mathcal{U}}_t$ by Eq.~\eqref{eq:ensemble_uncertainty_map}
                \ENDIF
                \STATE Compute $r_t$ by Eq.~\eqref{eq:ppo_total_reward}
                \STATE Refresh the sensing target
                \STATE Store rollout data and episode-end flags
                \ENDFOR
                \STATE Compute GAE advantages and returns
                \STATE Update $\vartheta_{\rm a}$, $\vartheta_{\rm c}$ by Eq.~\eqref{eq:ppo_training_loss}
                \ENDFOR
            \end{algorithmic}
        \end{minipage} \\
        \bottomrule
    \end{tabular}
    \endgroup
\end{table}

Algorithm~\ref{alg:air_ground_training} summarizes UAV policy training with a pretrained reconstruction ensemble and UGV controller $\mu$. Each episode starts with an initial map reconstructed from pre-mission measurements $\mathcal{D}_0$, whose uncertainty guides sensing target selection. The UAV then collects and transmits measurements with UGV support, while map updates at the UGV provide uncertainty feedback for subsequent decisions. After every $L$ interaction steps, GAE provides advantage and return estimates for PPO updates of the actor and critic. Training continues for $N_{\rm upd}$ policy updates.

\section{Experimental Evaluation}
\subsection{Dataset Setup}
\label{sec:datasets_setup}

We construct the PSD radio map dataset using large-scale propagation fields from the RadioMapSeer and ARM-Omni datasets. RadioMapSeer provides urban building layouts and physics-based radio propagation maps \cite{levie2021radiounet}, whereas ARM-Omni provides low-altitude aerial radio maps \cite{gao2026farm}. The propagation fields are cropped into $100\times100$ grid scenes.

\begin{figure}[!ht]
    \centering
    \includegraphics[width=0.96\columnwidth]{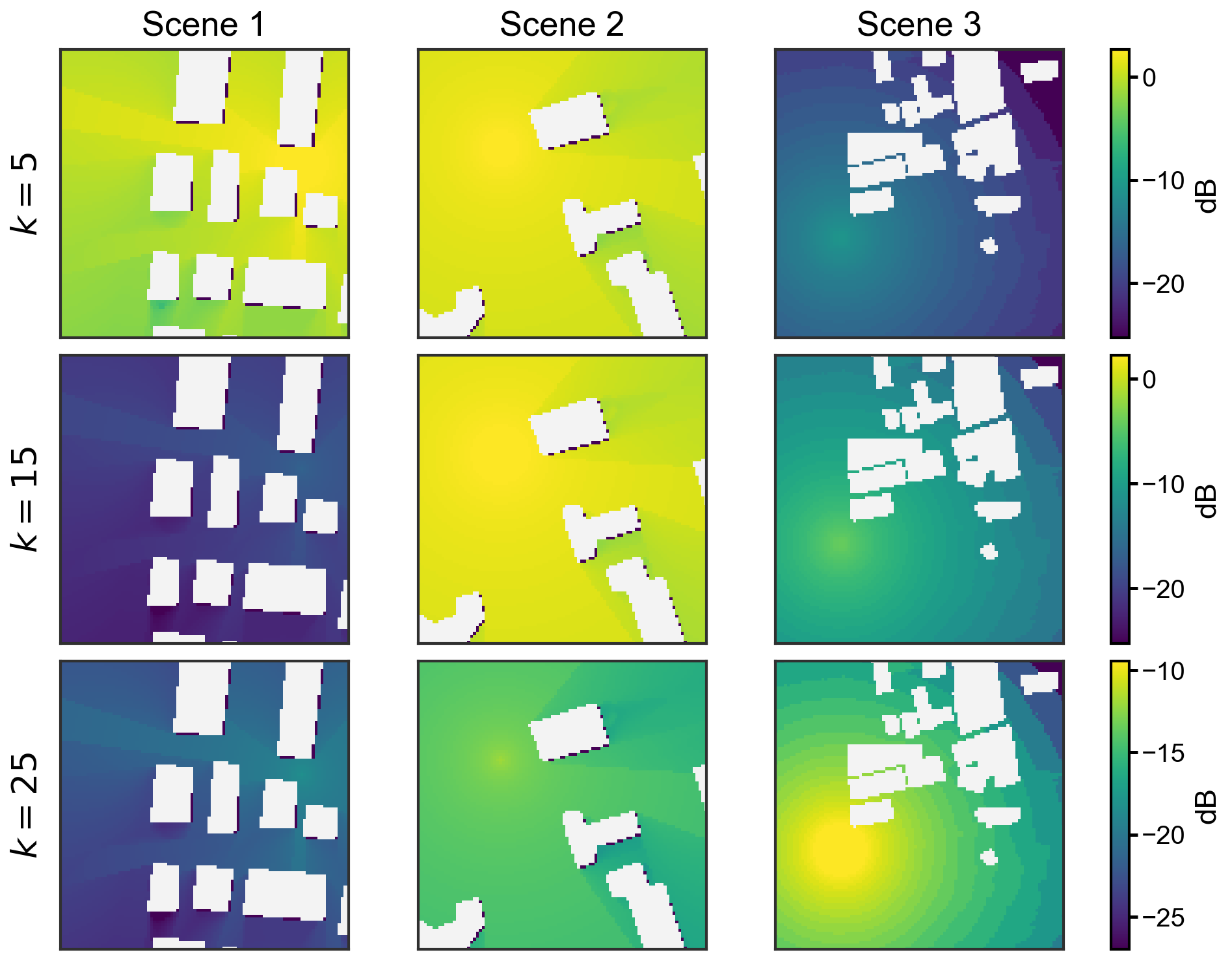}
    \caption{Representative ground-truth PSD maps from the constructed dataset. Columns show three scenes and rows show $k=5,15,25$.}
    \label{fig:dataset_ground_truth_examples}
\end{figure}

Following the TD model used in spectrum cartography \cite{zhang2020spectrum,chen2023offgrid,timilsina2024quantized,sun2024iibtd}, we independently synthesize eight nonnegative power spectra $\boldsymbol{\phi}\in\mathbb{R}_{+}^{K}$ for each spatial field, with $K=30$. Each spectrum is formed from squared sinc components and normalized such that $\sum_{k=1}^{K}\phi_k=K$. Pairing each spectrum with a spatial field according to Eq.~\eqref{eq:btd_map_model} yields eight dense $100\times100\times30$ PSD tensors per base scene. This construction preserves the source-field attenuation and blockage patterns while introducing diverse frequency-domain structures. Representative slices are shown in Fig.~\ref{fig:dataset_ground_truth_examples}. The resulting dataset contains 1,504 base scenes and 12,032 PSD tensors. Base scenes are divided into training, validation, and test subsets in an $80\%/10\%/10\%$ ratio. Tables~\ref{tab:environment_mobility_settings} and~\ref{tab:training_settings} summarize the simulation and training settings.

\begin{table}[!t]
    \centering
    \setlength{\abovecaptionskip}{2pt}
    \caption{Simulation Parameters.}
    \label{tab:environment_mobility_settings}
    \label{tab:sensing_communication_settings}
    \begingroup
    \footnotesize
    \setlength{\tabcolsep}{2.0pt}
    \renewcommand{\arraystretch}{0.94}
    \begin{tabular*}{\columnwidth}{@{\extracolsep{\fill}}lc@{\hspace{8pt}}lc@{}}
        \toprule
        \multicolumn{2}{c}{Environment and mobility} & \multicolumn{2}{c}{Sensing and communication} \\
        \cmidrule(lr){1-2}\cmidrule(lr){3-4}
        Parameter & Value & Parameter & Value \\
        \midrule
        Grid spacing $\Delta_{\rm cell}$ & 2 m/cell & Bandwidth $B$ & 100 MHz \\
        Building height & 25 m & Bandwidth units $N$ & 12 \\
        UAV altitude $h_{\rm UAV}$ & 30 m & Carrier frequency $f_{\rm c}$ & 3.5 GHz \\
        UAV step size & 4 cells & Outage threshold $\Upsilon_{\rm out}^{\rm dB}$ & $-5$ dB \\
        UGV step size & 5 cells & Buffer capacity $Q_{\max}$ & 512 Mbits \\
        Energy budget $E_{\max}$ & 9 kJ & Payload/band $G_{\rm band}$ & 8 Mbits \\
        Flight energy/cell $E_{\rm fly}$ & 12 J & Max sensing energy $E_{\rm sen}$ & 5 J \\
        Hovering energy $E_{\rm hov}$ & 8 J & &  \\
        \bottomrule
    \end{tabular*}
    \endgroup
\end{table}

\begin{table}[!t]
    \centering
    \setlength{\abovecaptionskip}{2pt}
    \caption{Training Hyperparameters.}
    \label{tab:training_settings}
    \begingroup
    \footnotesize
    \setlength{\tabcolsep}{2.0pt}
    \renewcommand{\arraystretch}{0.94}
    \begin{tabular*}{\columnwidth}{@{\extracolsep{\fill}}lc@{\hspace{8pt}}lc@{}}
        \toprule
        \multicolumn{2}{c}{ODU-TD training} & \multicolumn{2}{c}{PPO training} \\
        \cmidrule(lr){1-2}\cmidrule(lr){3-4}
        Parameter & Value & Parameter & Value \\
        \midrule
        Unfolded stages $T_{\rm DU}$ & 3 & Optimizer & Adam \\
        Optimizer & AdamW & Learning rate & $10^{-4}$ \\
        Learning rate & $10^{-4}$ & Discount factor $\gamma$ & 0.99 \\
        Batch size & 16 & GAE parameter $\lambda_{\rm GAE}$ & 0.95 \\
        Training epochs & 150 & Clip coefficient $\epsilon_{\rm PPO}$ & 0.2 \\
        & & Episode horizon $T$ & 160 \\
        & & Epochs per update & 6 \\
        \bottomrule
    \end{tabular*}
    \endgroup
\end{table}

\begingroup
\setlength{\textfloatsep}{8pt plus 2pt minus 2pt}
\setlength{\floatsep}{6pt plus 2pt minus 2pt}
\subsection{Online Spectrum Cartography Evaluation}

To evaluate the spectrum cartography algorithms in an online setting, we simulate sequential measurement arrivals by adding UAV measurements one at a time. Offline TD, online TD, and ODU-TD receive the same ordered measurements and reconstruct the same PSD map after each addition. We record NMSE after every update and measure module-wise and cumulative solver times over the complete sequence.

\begin{table}[!t]
    \centering
    \setlength{\abovecaptionskip}{2pt}
    \caption{Average runtime per reconstruction update.}
    \label{tab:runtime_update}
    \begingroup
    \footnotesize
    \setlength{\tabcolsep}{2.5pt}
    \begin{tabular*}{\columnwidth}{@{\extracolsep{\fill}}lcccc@{}}
        \toprule
        Method     & $\boldsymbol{\Theta}_{ij}$ (ms) & $\boldsymbol{\Phi}$ (ms) & $\mathbf{S}_r$ (ms) & Total (ms) \\
        \midrule
        ODU-TD          & \textbf{1.43}          & \textbf{1.19}        & \textbf{0.73}       & \textbf{3.35}       \\
        online TD       & 1.45          & 1.07        & 183.16     & 185.69     \\
        offline TD      & 15.13         & 2.34        & 403.51     & 420.98     \\
        \bottomrule
    \end{tabular*}
    \endgroup
\end{table}

Table~\ref{tab:runtime_update} shows that repeated SVT in the update of $\mathbf{S}_{r}$ dominated online TD, requiring about $183$ ms. Replacing SVT with the learned proximal module reduced a complete update from $185.69$ ms to $3.35$ ms. The comparable update times for $\boldsymbol{\Theta}$ and $\boldsymbol{\Phi}$ confirm that selectively replacing the $\mathbf{S}_r$ update removes the dominant bottleneck.

Fig.~\ref{fig:iibtd_solver_eval} shows that ODU-TD reduced NMSE more rapidly during early sampling and maintained a smoother trajectory below both baselines. Initially, limited measurements and spatial coverage leave local fits in offline TD and online TD insufficiently constrained. Newly added samples can substantially change these fits and affect estimates in less-observed regions through shared spectral factors and low-rank coupling, potentially causing temporary NMSE increases. In contrast, ODU-TD combines new measurement information with spatial structure captured in previous estimates, reducing sensitivity to individual sample additions. Its learned convolutional prior also captures local patterns beyond global nuclear-norm regularization, supporting both more stable updates and lower reconstruction error. Overall, ODU-TD required $1.58$ s, compared with $42.87$ s for online TD, achieving an approximately $27$-fold speedup by replacing repeated SVT with a fixed-depth convolutional update.

\begin{figure}[t]
    \centering
    \includegraphics[width=\linewidth]{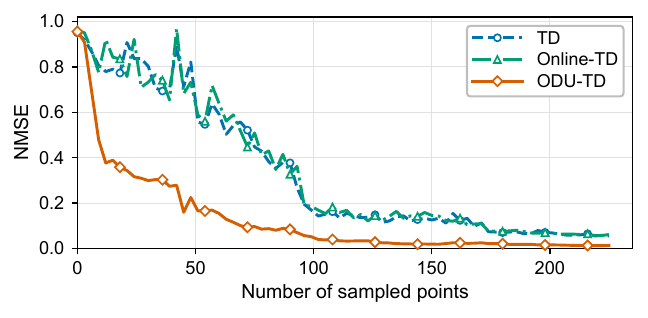}
    \caption{NMSE of spectrum cartography methods after each sample addition under identical ordered UAV measurements.}
    \label{fig:iibtd_solver_eval}
\end{figure}

Fig.~\ref{fig:solver_tensor_slices} compares the ground-truth and reconstructed PSD map slices using identical sampling locations across methods. ODU-TD achieved a final NMSE of $0.0063$, compared with $0.0259$ for offline TD and $0.0270$ for online TD. Both baselines produced stripe artifacts and blurred localized attenuation near building boundaries. ODU-TD better preserved these structures and exhibited fewer stripe-like artifacts.

\begin{figure}[t]
    \centering
    \includegraphics[width=\linewidth]{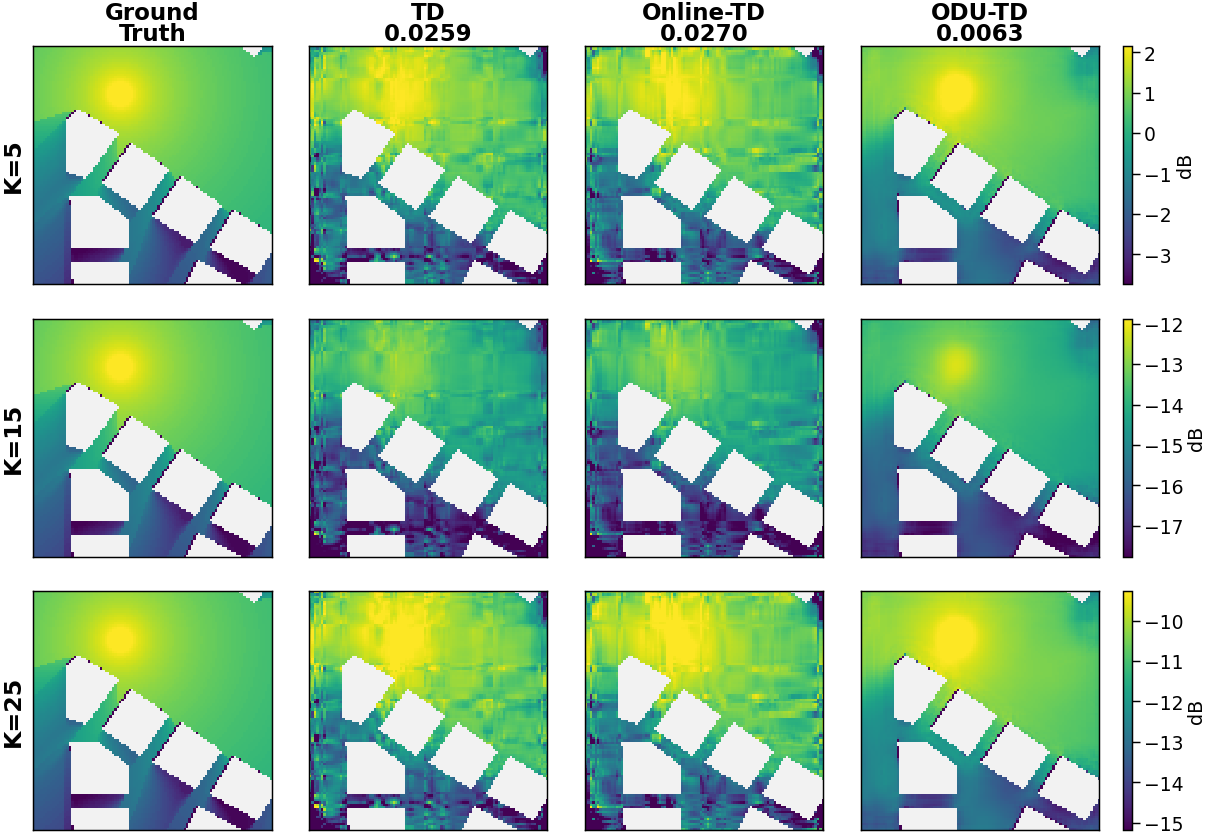}
    \caption{Representative PSD radio map slices at $k=5,15,25$ using identical sampling locations across methods.}
    \label{fig:solver_tensor_slices}
\end{figure}

\subsection{Air--Ground Cooperation Policy Evaluation}
\subsubsection{Compared Methods}

\noindent\textbf{Fixed-UGV:} The UGV remains at one support location, while the UAV uses a separately trained policy with the same actor architecture.

\noindent\textbf{Greedy-A*:} The UAV uses uncertainty-guided, non-learning motion with heuristic sensing and quantization choices. The UGV uses predictive dual-path A* support without recovery switching.

\noindent\textbf{Target-Only A*:} UAV-only PPO is paired with direct UGV routing toward the sensing target. This baseline omits predictive link scoring and recovery switching.

\noindent\textbf{MAPPO and IPPO:} Multi-agent PPO and independent PPO both learn the actions of the UAV and UGV~\cite{yu2022mappo,de2020ippo}.

\noindent\textbf{HAPPO:} Heterogeneous-agent PPO applies sequential policy updates to the UAV and UGV~\cite{kuba2022happo,deng2026happo}.

\noindent\textbf{Random:} Feasible actions are sampled uniformly for both the UAV and UGV.

\begin{figure}[t]   
    \centering
    \setlength{\abovecaptionskip}{2pt}
    \includegraphics[width=\columnwidth]{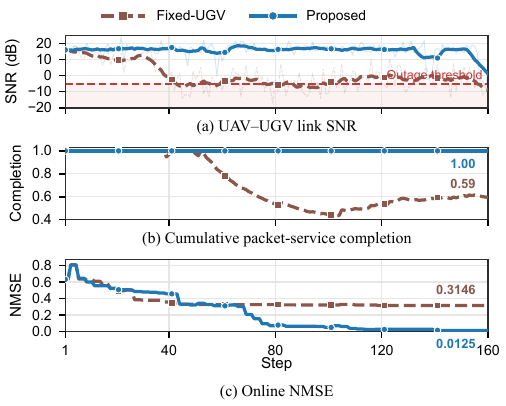}
    \caption{Slot-wise comparison of Fixed-UGV and the proposed framework in one episode with quantized transmission.}
    \label{fig:mobile_ugv_cooperation}
\end{figure}

\subsubsection{Evaluation Metrics}
\label{sec:evaluation_metrics}

The evaluation uses three complementary metrics to characterize reconstruction accuracy, finite-horizon measurement service, and air--ground link reliability.

\textbf{NMSE:} The NMSE in Eq.~\eqref{eq:nmse_def} measures the accuracy of the reconstructed PSD radio map and is therefore the primary task-level metric. A lower NMSE indicates more faithful recovery, whereas a higher value indicates worse reconstruction performance.

\textbf{Packet-service completion:} Let $B_{\rm prod}=\sum_{t=1}^{T}G_t$ be the total generated payload and $B_{\rm cmp}$ the payload contained in packets fully delivered by the mission end. The packet-service completion ratio is
\begin{equation}
    \eta_{\rm cmp}
    =
    \begin{cases}
        \dfrac{B_{\rm cmp}}{B_{\rm prod}}, & B_{\rm prod}>0,\\[4pt]
        1, & B_{\rm prod}=0
    \end{cases},
    \label{eq:completion_ratio}
\end{equation}
where $\eta_{\rm cmp}=1$ by convention when no payload is generated. For $B_{\rm prod}>0$, it measures the fraction of generated payload that is fully delivered and usable before the mission ends. It therefore evaluates how effectively sensing, quantization, and channel service are coordinated over the finite mission horizon. A high ratio indicates that the measurement workload generated by the selected sensing ratio and quantization bit depth can be supported by the available air--ground link. Conversely, a low ratio reveals a mismatch between the generated workload and the link service capacity, causing queue accumulation at the end of the mission.

\textbf{Outage ratio:} The fraction of mission slots in outage is
\begin{equation}
    \eta_{\rm out}
    =
    \frac{1}{T}\sum_{t=1}^{T}
    \mathbb{I}\!\left\{
    \Upsilon_t^{\rm dB}
    (\boldsymbol{u}_t,\boldsymbol{g}_t,\rho_t)
    <\Upsilon_{\rm out}^{\rm dB}
    \right\}.
    \label{eq:outage_ratio}
\end{equation}
A lower outage ratio indicates more persistent air--ground connectivity and more stable opportunities for measurement transmission. A higher ratio indicates recurrent service interruptions, increasing the risk of queue accumulation and packet loss. This metric is important because sustained link availability preserves the measurement flow required for continuous online map updates.

\begin{figure}[t]
    \centering
    \setlength{\abovecaptionskip}{2pt}
    \includegraphics[width=\columnwidth]{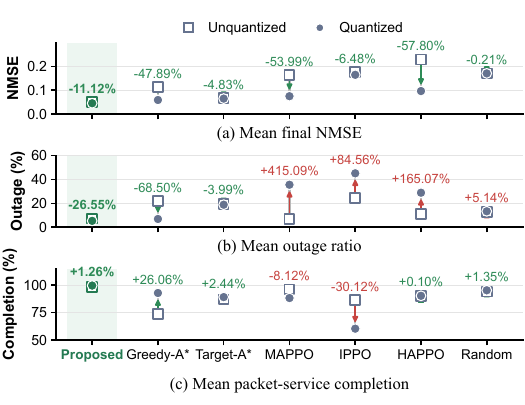}
    \caption{Paired performance of seven policies under unquantized (open squares) and quantized (filled circles) transmission.}
    \label{fig:policy_nmse_service_tradeoffs}
\end{figure}

\subsubsection{Mobile versus Fixed UGV Support}
Fig.~\ref{fig:mobile_ugv_cooperation} compares Fixed-UGV with the proposed framework in a representative episode with quantized transmission. Repeated Fixed-UGV outages produced a $21.88\%$ outage ratio, reduced packet-service completion to $59.19\%$, and left the final NMSE at $0.3146$. The proposed framework instead maintained zero outage and $100\%$ completion, allowing successive reconstruction updates and reducing the final NMSE to $0.0125$. The coupled changes in NMSE and completion show that mobile support improved reconstruction by sustaining measurement service as the UAV moved through blocked regions.

\begingroup
\setlength{\intextsep}{4pt}
\begin{table}[!htb]
    \centering
    \setlength{\abovecaptionskip}{2pt}
    \caption{Mean policy performance.}
    \label{tab:quantized_policy_comparison}
    \begingroup
    \footnotesize
    \setlength{\tabcolsep}{3.0pt}
    \renewcommand{\arraystretch}{1.06}
    \begin{tabular*}{\columnwidth}{@{\extracolsep{\fill}}lccc@{}}
        \toprule
        Method & NMSE $\downarrow$ & Outage (\%) $\downarrow$ & Completion (\%) $\uparrow$ \\
        \midrule
        Proposed          & \textbf{0.0454} & \textbf{5.26} & \textbf{99.62} \\
        Greedy-A*         & 0.0591 & 6.85 & 92.87 \\
        Target-Only A*    & 0.0649 & 18.80 & 89.16 \\
        MAPPO             & 0.0750 & 35.55 & 88.31 \\
        IPPO              & 0.1646 & 45.13 & 60.47 \\
        HAPPO             & 0.0966 & 28.85 & 90.34 \\
        Random            & 0.1706 & 13.31 & 95.34 \\
        \bottomrule
    \end{tabular*}
    \endgroup
\end{table}

\endgroup

\subsubsection{Impact of Quantized Transmission}
Fig.~\ref{fig:policy_nmse_service_tradeoffs} compares each policy before and after quantization. For the proposed framework, quantization reduced mean NMSE from $0.0511$ to $0.0454$ and outage from $7.16\%$ to $5.26\%$, while completion increased from $98.38\%$ to $99.62\%$. Greedy-A* exhibited the same favorable direction across all three metrics. By contrast, the multi-agent policies achieved lower NMSE but incurred higher outage, while MAPPO and IPPO also showed lower packet-service completion ratios. Thus, the reduced payload improved end-to-end performance only when the mobility strategy converted the released transmission capacity into sustained measurement service.

\subsubsection{Overall Performance}
Table~\ref{tab:quantized_policy_comparison} shows that the proposed framework achieved the lowest NMSE of $0.0454$ and outage ratio of $5.26\%$, both over $23\%$ lower than Greedy-A*. It also attained the highest completion ratio of $99.62\%$, a relative improvement of $4.49\%$ over the best baseline. Together, these results demonstrate that the framework achieves accurate online spectrum cartography while maintaining reliable links and delivering nearly all generated measurement payload in the evaluated blockage-prone urban scenarios.

\FloatBarrier

\endgroup

\section{Conclusions}
This paper developed an air--ground cooperative framework for active online spectrum cartography under UAV resource constraints and urban blockage. Deep-unfolded reconstruction enables timely map updates, while interpolation-error analysis guides quantization bit-depth selection to balance measurement fidelity and transmission load. Through an uncertainty-guided feedback loop, the framework coordinates simultaneous UAV sensing and transmission with mobile UGV communication and computing support. Experiments show that the proposed reconstruction method effectively accelerated online map updates. In urban scenes, the proposed framework demonstrated lower reconstruction error, fewer communication outages, and more complete measurement delivery than the evaluated baselines. Future work will address unknown, time-varying environments and extend the framework to support cooperation among multiple UAVs and UGVs.

\FloatBarrier

\appendices
\section{Proof of Proposition~\ref{prop:log_quant_error}}
\label{app:proof_log_quant_error}

\begin{IEEEproof}
Using the same notation as in Proposition~\ref{prop:log_quant_error}, let $\Delta=\Delta_z$ and define the log-domain quantization error as
$q\triangleq\widehat z_m-z_m$. Using
$\widehat z_m=z_m+q$, the inverse logarithmic transformation gives
$\widehat y_m=\left(y_m+\epsilon_0\right)e^{q}-\epsilon_0$.
Therefore, the linear-domain quantization error is  $ e_{q,m}=\left(y_m+\epsilon_0\right)(e^{q}-1).$

Under the pseudo-quantization-noise model, $q$ is independent of
$y_m$ and uniform over $[-\Delta/2,\Delta/2]$. Its first
exponential moment is $\mathbb E[e^{q}]
    =\frac{1}{\Delta}\int_{-\Delta/2}^{\Delta/2}e^{u}\,\mathrm du
    =\frac{2\sinh(\Delta/2)}{\Delta}.$
Similarly, $\mathbb E[e^{2q}]=\sinh(\Delta)/\Delta$.
The Taylor expansion
$\sinh(x)=x+x^3/6+O(x^5)$ at $x=0$ gives
$\mathbb E[e^{q}]=1+\Delta^2/24+O(\Delta^4)$. Since $q$ is independent of $y_m$, taking the conditional expectation of $e_{q,m}$ gives $\mathbb E\!\left[e_{q,m}\mid y_m\right]
    =\left(y_m+\epsilon_0\right)
    \left\{\frac{\Delta^2}{24}+O(\Delta^4)\right\}.$
By the preceding definition, the centered quantization error $\widetilde e_{q,m}(b)$ is given by $ \widetilde e_{q,m}(b)={}
    \left(y_m+\epsilon_0\right)
    \left\{e^{q}-\mathbb E[e^{q}]\right\}.$
By construction, its conditional mean is $\mathbb E[\widetilde e_{q,m}(b)\mid y_m]=0$.
The conditional variance can be expressed as
\begin{equation}
    \operatorname{Var}\!\left[
        \widetilde e_{q,m}(b)\mid y_m
    \right]
    =\left(y_m+\epsilon_0\right)^2
    \left[\frac{\Delta^2}{12}+O(\Delta^4)\right].
\label{eq:centered_log_quant_variance_proof}
\end{equation}
\end{IEEEproof}

\section{Proof of Proposition~\ref{prop:quant_iibtd_error}}
\label{app:proof_quant_iibtd_error}

\begin{IEEEproof}
We restore the frequency-band index $k$ for the joint analysis across bands. Assume that log-domain quantization errors are mutually independent across measurements and frequency bands and independent of the unquantized observations. For the centered errors in Proposition~\ref{prop:log_quant_error}, averaging the conditional variances gives the covariance in band $k$:
\begin{equation*}
    \boldsymbol\Sigma_{q,k}(b)
    \triangleq
    \operatorname{diag}\!\left(
    \left\{
    \mathbb{E}\!\left[
    \operatorname{Var}\!\left(
    \widetilde e_{q,m}^{(k)}(b)\mid y_m^{(k)}
    \right)\right]
    \right\}_m
    \right).
\end{equation*}

Following \cite{sun2024iibtd}, consider a single-source local estimate at grid cell $(i,j)$ with $\nu=0$ and a nonsingular normal matrix. Sampling locations, observed bands, spatial weights, source spectrum, and the unquantized measurement distribution are fixed as $b$ varies. Let $\Omega$ denote the observed band set. Suppressing the slot index $t$, define $\mathbf A_1=\mathbf X_{ij}\mathbf Q_{ij}^{2}\mathbf X_{ij}^{\top}$ and $\mathbf A_2=\mathbf X_{ij}\mathbf Q_{ij}^{4}\mathbf X_{ij}^{\top}$. 
Let $\widetilde{\boldsymbol r}_{ij}$ stack the spectrally scaled fading, receiver noise, and centered quantization error over $k\in\Omega$. With $\boldsymbol\phi=[\phi_k]_{k\in\Omega}^{\top}$, the noise term in the local normal equations is $\boldsymbol\xi_{ij}\triangleq(\boldsymbol\phi^{\top}\otimes\mathbf X_{ij}\mathbf Q_{ij}^{2})\widetilde{\boldsymbol r}_{ij}$. The quantization component embedded in
$\widetilde{\boldsymbol r}_{ij}$ contributes
\begin{equation}
    \begin{aligned}
    \mathbf C_{q,ij}(b)
    \triangleq
    \sum_{k\in\Omega}\phi_k^2
    \mathbf X_{ij}\mathbf Q_{ij}^{2}
    \boldsymbol\Sigma_{q,k}(b)
    \mathbf Q_{ij}^{2}\mathbf X_{ij}^{\top}.
    \end{aligned}
    \label{eq:weighted_quantization_covariance}
\end{equation}
Assuming that the fading, receiver-noise, and centered quantization errors
are mutually uncorrelated gives
\begin{equation}
    \begin{aligned}
    \mathbb E(\boldsymbol\xi_{ij}\boldsymbol\xi_{ij}^{\top})
    ={}&
    \left(
    \sigma_\eta^2\sum_{k\in\Omega}\phi_k^4
    +\sigma_\epsilon^2\sum_{k\in\Omega}\phi_k^2
    \right)\mathbf A_2 +\mathbf C_{q,ij}(b).
    \end{aligned}
    \label{eq:iibtd_weighted_residual_covariance}
\end{equation}
The local coefficient estimate $\widehat{\boldsymbol\theta}_{ij}^{r}(b)$ has the normal matrix
    $\sum_{k\in\Omega}\phi_k^2\mathbf A_1$. Propagating
Eq.~\eqref{eq:iibtd_weighted_residual_covariance} through this estimator
yields
\begin{equation}
    \operatorname{Cov}\!\left[
    \widehat{\boldsymbol\theta}_{ij}^{r}(b)
    \right]
    =
    \frac{
    \mathbf A_1^{-1}
    \mathbb E(\boldsymbol\xi_{ij}\boldsymbol\xi_{ij}^{\top})
    \mathbf A_1^{-1}
    }{
        \left(\sum_{k\in\Omega}\phi_k^2\right)^2
    }.
    \label{eq:iibtd_parameter_covariance}
\end{equation}
Taking expectations in Eq.~\eqref{eq:prop1_log_error_variance_scaling} gives
\begin{equation}
    \begin{aligned}
    \boldsymbol\Sigma_{q,k}(b)
    ={}&\left\{\frac{\Delta_z^2}{12}
    +O\!\left[\Delta_z^4\right]\right\}\\
    &\times\operatorname{diag}\!\left(
    \left\{\mathbb{E}\!\left[\left(y_m^{(k)}+\epsilon_0\right)^2\right]\right\}_m
    \right).
    \end{aligned}
    \label{eq:local_quantization_covariance_expansion}
\end{equation}
Substituting Eq.~\eqref{eq:local_quantization_covariance_expansion} into
Eq.~\eqref{eq:weighted_quantization_covariance} and then into
Eq.~\eqref{eq:iibtd_parameter_covariance} yields
\begin{equation}
    \operatorname{Cov}\!\left[
    \widehat{\boldsymbol\theta}_{ij}^{r}(b)
    \right]
    \approx
    \mathbf C_{0,ij}
    +\Delta_z^2\mathbf C_{\Delta,ij},
    \label{eq:iibtd_parameter_covariance_scaling}
\end{equation}
where the bit-depth-independent matrices $\mathbf C_{0,ij}\succeq\mathbf0$ and
$\mathbf C_{\Delta,ij}\succeq\mathbf0$ capture the fading and receiver-noise
contribution and the leading quantization contribution, respectively.

Because the reference value is deterministic and the local polynomial basis is centered at $(i,j)$, the interpolation-error variance equals the first diagonal entry of the covariance in Eq.~\eqref{eq:iibtd_parameter_covariance_scaling}. Using $\Delta_z=(z_{\max}-z_{\min})/(2^b-1)$, we obtain
\begin{equation*}
    \operatorname{Var}\!\left[e_{ij}^{r}(b)\right]
    \approx a(2^b-1)^{-2}+c,
\end{equation*}
where $c=[\mathbf C_{0,ij}]_{1,1}$ and
$a=(z_{\max}-z_{\min})^2[\mathbf C_{\Delta,ij}]_{1,1}\geq0$.
Under the same fixed-design and centered-noise assumptions, nonsingular multi-source block updates change only the bit-depth-independent coefficient matrices, preserving this scaling for each source component.
\end{IEEEproof}

\bibliographystyle{IEEEtran}
\bibliography{refs}

\end{document}